%% file: main.tex
\documentclass[10pt,journal,compsoc]{IEEEtran}
\ifCLASSOPTIONcompsoc
  \usepackage[nocompress]{cite}
\else
  \usepackage{cite}
\fi

\ifCLASSINFOpdf
   \usepackage[pdftex]{graphicx}
   \DeclareGraphicsExtensions{.pdf,.jpeg,.png}
\else
   \usepackage[dvips]{graphicx}
   \DeclareGraphicsExtensions{.eps}
\fi

\usepackage{algorithmic}
\usepackage{algorithm}
\usepackage{cite}
\usepackage{wrapfig}
\usepackage{siunitx}
\usepackage{subfig}
\usepackage{balance}
\usepackage{amsmath,amssymb,amsfonts} 
\usepackage{color}
\usepackage{enumitem}
\usepackage{ltxtable}
\usepackage[table]{xcolor}  
\usepackage{multirow}

\usepackage{algorithm}
\usepackage{amsthm}
\usepackage{multirow}
\usepackage{float}
\usepackage{caption}

\newtheorem{definition}{Definition}
\newtheorem{theorem}{Theorem}

\usepackage{color}
\definecolor{mycolor}{rgb}{1,0.0,0.5}

\begin{document}
\title{Strategic Defense against Stealthy Link Flooding Attacks: A Signaling Game Approach}

\author{Abdullah Aydeger, \IEEEmembership{Student Member, IEEE},
        Mohammad Hossein Manshaei, \IEEEmembership{Member, IEEE},\\
        Mohammad Ashiqur Rahman, \IEEEmembership{Member, IEEE}, Kemal Akkaya, \IEEEmembership{Senior Member, IEEE}
\IEEEcompsocitemizethanks{\IEEEcompsocthanksitem The authors (respective emails: \{aayde001, mmanshae, marahman, kakkaya\}@fiu.edu) are with the Department of Electrical and Computer Engineering, Florida International University, Miami, FL 33174, USA. 

\IEEEcompsocthanksitem A. Aydeger and M. Manshaei are the co-first authors. 

\IEEEcompsocthanksitem A. Aydeger is the corresponding author.}}


\IEEEtitleabstractindextext{%

\begin{abstract}
With the increasing diversity of Distributed Denial-of-Service (DDoS) attacks, it is becoming extremely challenging to design a fully protected network. 
For instance, Stealthy Link Flooding Attack (SLFA) is a variant of DDoS attacks that strives to block access to a target area by flooding a small set of links, and it is shown that it can bypass traditional DDoS defense mechanisms. 
One potential solution to tackle such SLFAs is to apply Moving Target Defense (MTD) techniques in which network settings are dynamically changed to confuse/deceive attackers, thus making it highly expensive to launch a successful attack.
However, since MTD comes with some overhead to the network, to find the best strategy (i.e., when and/or to what extent) of applying it has been a major challenge. The strategy is significantly influenced by the attacker's behavior that is often difficult to guess. 
In this work, we address the challenge of obtaining the optimal MTD strategy that effectively mitigates SLFAs while incurs a minimal overhead. We design the problem as a signaling game considering the network defender and the attacker as players.  
A belief function is established throughout the engagement of the attacker and the defender during this SLFA campaign, which is utilized to pick the best response/action for each player. 
We analyze the game model and derive a defense mechanism based on the equilibria of the game. We evaluate the technique on a Mininet-based network environment where an attacker is performing SLFAs and a defender applies MTD based on equilibria of the game. The results show that our signaling game-based dynamic defense mechanism can provide a similar level of protection against SLFAs like the extensive MTD solution, however, causing a significantly reduced overhead.
\end{abstract}

\begin{IEEEkeywords}
Stealthy Link Flooding Attack, Crossfire Attack, Signaling Game, Moving Target Defense
\end{IEEEkeywords}}


\maketitle
\begingroup
\renewcommand{\cleardoublepage}{}
\renewcommand{\clearpage}{}

\input{introduction}
\input{preliminaries}

\section{Crossfire Attack and Defense Game Model}
\label{sec:game}

In this section, we first give an overview of the proposed game. Then, we define the players' action sets and model the belief function and payoffs.

\subsection{Overview} \label{sec:overview}


In our 
Crossfire Attack and Defense Game model, a client can be malicious (a bot corresponding to a Crossfire attacker) or benign. Clients send various requests to the servers in the target network. A bot makes ping and traceroute requests so that the attacker can do the reconnaissance (i.e., link map construction) to launch a Crossfire attack. 
The security administrator of the target network (simply the defender) can apply RRM to thwart this reconnaissance, so  the data phase of the Crossfire attack. In the case of applying RRM, the defender randomly changes/mutates the routing paths, thus changing the link map.
We model the interaction between the Crossfire attacker (in fact, each bot individually) 
and the defender using a \emph{signaling game} $\mathcal{G}^{cf}$. Signaling game is a two-player incomplete information game, in which Nature has a unique randomizing strategy (i.e., $\theta$) that is commonly known to both the defender and the attacker. With this randomizing strategy the type of sender would be defined. 




The first player, a.k.a. sender (here, the attacker/bot), is informed of Nature's choice and chooses an action. The second player, a.k.a. the receiver (here, the defender), then chooses an action without knowing Nature's choice but observing the first player's action. Our choice of the signaling game is based on the dynamic and incomplete information characteristic of the Crossfire attack, where the action of one player is conditioned over its belief about the type of the opponent. The action flow of the game is shown in Fig.~\ref{fig:signaling-game}.


The game $\mathcal{G}^{cf}$ is played individually with each client/sender who can be a bot or a legitimate user. 
A bot 
is considered as a type $t_1$ sender, while a legitimate user as a type $t_2$ sender. We represent this set of sender types as $\mathbb{T}$, i.e., $\mathbb{T} = \{t_1, t_2\}$.
The second player of game $\mathcal{G}^{cf}$ is the defender. 
The game is played in the following steps, as shown in Fig.~\ref{fig:signaling-game}. The defender receives two different types of messages, i.e., \emph{Reconnaissance} or \emph{Regular Traffic}. Nature draws type $t_1$ or $t_2$ with a probability of $\theta$ or $1-\theta$, respectively. The defender does not know exactly whether the observed message is coming from a bot or a legitimate user but s/he can only form/follow a belief. According to this belief, the defender must decide whether to defend by applying RRM or not. 
In the following subsections, we define the actions of the players, and accordingly model the belief and payoff functions. In the game, we often use the term "attacker" to represent a bot. 

\begin{figure*}[ht]
\centering
\includegraphics[scale=0.4]{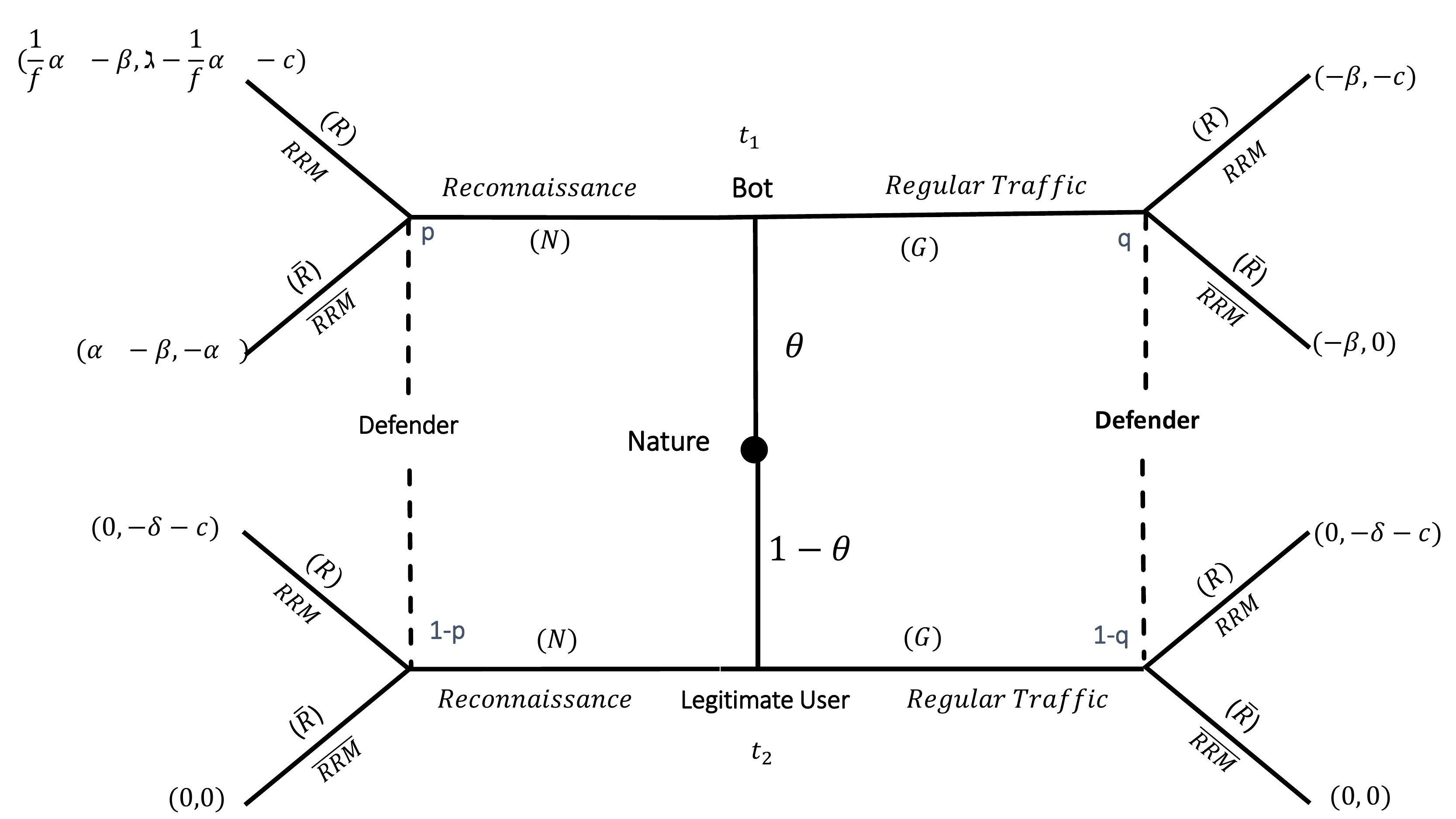}
\caption{$\mathcal{G}^{cf}$ signaling game model representation.}
\label{fig:signaling-game}
\vspace{-14pt}
\end{figure*}

\begin{table}[t]
\centering
\caption{Symbols used in manuscript}
\rowcolors{2}{gray!35}{white}
\label{Table:Symbols}
\begin{tabular}{c | l }    
\textbf{Symbol} & \textbf{\mbox{} Definition} \\  
\hline
\hline
N & Reconnaissance packet sender     \\
G & Regular packet sender     \\
$\overline{R}$ & Defender's action as no RRM     \\
$R$ & Defender's action as RRM     \\
$\alpha$ & Attacker's gain from single reconnaissance packet \\
$\theta$ & Belief value of whether the client is legitimate or bot \\
$\beta$ & Cost for renting bots \\ 
$\delta$ & Cost of not providing information to legitimate user \\ 
$\lambda$ & Gain of defending through reconnaissance  \\
$c$ & Cost of applying RRM \\
$f$ & Frequency of Random Route Mutation (RRM) \\
$N_c$ & Number of hosts in the network

\end{tabular}
\vspace{-14pt}
\end{table}

\subsection{Action Sets} 
\label{sec:strategymodel}

In game $\mathcal{G}^{cf}$, we define the action set $\mathbb{A}$ of the first player by an ordered pair $(m(t_1),m(t_2))$, where $m(t_1)$ is the action of type $t_1$ (i.e., a bot) and $m(t_2)$ is the action of type $t_2$ (i.e., a legitimate user). We also assume that each sender, irrespective of his/her type, can select his/her action from the same action set. Let $\mathbb{A}$ be $\{N,G\}$, where $N$ represents sending reconnaissance packets while $G$ is about sending the usual/regular traffic. It is worth noting that if the attacker sends reconnaissance packets (i.e., plays $N$), then the defender may become suspicious. 
On the other hand, when the attacker sends regular traffic (i.e., plays $G$), 
s/he will not gain necessary information. 

By using these actions, the attacker can make a strategic plan with a combination of actions. This strategy could either be slow and stealthy, or quick and greedy. If the attacker's strategy is stealthy one as defined in Section \ref{sec:crossfire}, a longer time will required to do the reconnaissance. In this long period, there is a higher possibility to have changes in the link map due to existing defense strategies (or even usual network events).
Hence, the attack will have less chance to be successful. On the other hand, if the attacker is so greedy about the attack, s/he might immediately do harm yet possibly end up being caught quickly. 
Thus, the attacker needs to do a trade-off between the time spent for reconnaissance and the risk of being detected by the defender. We will elaborate this issue in more detail in Section~\ref{sec:eval}. 

Similar to the sender, the receiver/defender has an action set $\mathbb{B}$. It 
can respond/reply with a route mutation to deceive the attacker (i.e., play $R$). Otherwise, it reply with no route change (i.e., play $\overline{R}$). Hence, $\mathbb{B} =\{R, \overline{R}\}$. 
If the defender applies route mutation (i.e., $R$), then there will be an extra cost depending on the time it takes to install flow rules on network devices. The cost for the regular replies (i.e., $\overline{R}$) is considered as zero, since there is not any additional efforts required.
We will later elaborate on the defense cost when we obtain the defender payoffs. 


\vspace{-10pt}
\subsection{Belief Model} 
\label{Sec:Belief}



In game $\mathcal{G}^{cf}$, the defender does not know whether it receives messages 
from an attacker or a legitimate user. Yet, it observes these messages and builds the belief function (i.e., $\theta$) about the sender's type 
based on this observation. 
Remember that this $\theta$ is a common knowledge between the sender and the receiver. 
The defender has a different belief for each client/sender: $\theta_y$, where $0 < y \leq N_c$. Furthermore, to represent the belief at a time instance $x$, we use the notation $\theta_y(x)$. The belief ranges from $0$ to $1$, while $0$ showing no sign of suspicion of the client's legitimacy. The belief has two parts: the initial belief that is statically calculated at the beginning of the game and the dynamic belief that is updated according to the sender's actions. 



\subsubsection{Initial Belief}

In our environment, we assume there is no information available about any client, and all clients are at the same level of suspicion. Thus, we assign 0 to each initial belief (i.e., $\theta_y(0)$). However, the defender can assign different values to the initial belief for each client if any prior information is available.




\subsubsection{Dynamic Belief}

The dynamic belief for a specific sender will be calculated each time the game is played. The belief is defined based on three different weighted factors, as shown in Equation~(\ref{Eq:belief}). While these factors are normalized values, the summation of their weights are 1 (100\%).
The first factor is the previous belief (i.e., $\theta_y(x-1)$) that is weighted by $F_1$. 
The reason is that if a client has been behaving suspicious or benign for most of the time, 
his/her current action, which may be different than this characteristic, cannot eliminate (or changes) his/her performance previously observed/learned.
That is why, we consider assigning a high weight (e.g., 90\% 
) to $F_1$.
%
\begin{equation}
\label{Eq:belief}
\theta_y(x) = \theta_y(x-1) \times F_1 + \frac {\sum_{k=1}^{N_c}  \theta_k(x-1)} {N_c} \times F_2 +  A(x) \times F_3
\end{equation}

The second factor is the average of the beliefs for all the network nodes (i.e., $\frac{\sum_{k=1}^{Nc}  \theta_k(x-1)} {Nc} $). This factor reflects the Crossfire attack characteristic in which a number of bots are used to perform the attack. 
$F_2$ is the weight representing its impact on the belief function and we consider it notable (e.g., 9\% 
). The last metric we use to update our belief is the current/latest packet that we have received from the sender (i.e., $A(x)$). It's weight is $F_3$, which we consider as small (e.g., 1\% 
. Even though its weight is considered small, it's impact on the belief adds up rapidly as the sender keeps sending packets to the receiver, and the game is continuously played for each packet receipt.



\subsection{Payoff Model}
\label{sub:Payoff}

In this section, we calculate the benefit and cost of each of the players considering all possible strategies.
Let us first consider the case where the game is played between the attacker/bot (i.e., the sender type is $t_1$) and the defender. This case is represented by the upper part of the signaling game as shown in Fig.~\ref{fig:signaling-game}.

If the defender applies RRM (i.e., plays $R$) in response to a reconnaissance message (i.e., the sender plays $N$) from the bot (i.e., given that Nature has chosen the bot to play with a probability of $\theta$), the defender can successfully defend the targeted server(s) against the attack. Let us assume that $\alpha$ represents the number of packet losses that the attacker can cause by leveraging the information that is collected through the reconnaissance process. In other words, $\alpha$ is the numerical gain of the attacker. If we assume that the frequency of running RRM, is $f$, then the benefit of the attacker will be degraded for the higher frequency of RRM.
Hence,in this case, we model the benefit of the attacker by $\frac{1}{f}\alpha$. We designate the attacker's cost by $\beta$, which is an increasing function on the number of bots that are deployed by the attacker in the network. This parameter also represents the capability of the attacker in our experiments: a decent or strong capability. 
While the decent attacker can afford a small number of bots, the strong attacker can deploy more bots. The number of packets sent by the attacker can be considered as another metric of the attacker's cost, yet this does not bring a direct cost to the attacker. 
Finally, considering the calculated benefit and cost for the attacker we conclude that the attacker's payoff is $\frac{1}{f}\alpha- \beta$. 

We represent the benefit of the defender (gain of defending against reconnaissance) by $\lambda$. This benefit is made by giving the attacker wrong information about reconnaissance, which prevents the data phase of the attack. This action brings protection to the network in terms of decreased packet loss. 
Moreover, we represent the total cost of applying RRM by $c$, which is measured in terms of additional delay and packet loss that might be caused due to the flow table update. 
Given the calculated cost and benefit for the defender, the payoff of the defender running RRM against the bot is: $\lambda - \frac{1}{f}\alpha -c$.

Following the above discussion for the payoffs of the attacker and defender, if the defender does not react to the reconnaissance messages of the bot, the payoff of the attacker and defender will be $-\alpha$ and $\alpha-\beta$, respectively. With a similar analysis, we can show that the payoff of the defender is $-c$. If it does not defend, the payoff would be $0$. The payoff of the attacker when s/he does not send reconnaissance packets will always be equal to $-\beta$.

Now we consider the second game, where the players are the legitimate user and the defender (i.e., the lower part of the signaling game presented in Fig.~\ref{fig:signaling-game}). Since RRM does not have significant impact on the legitimate users' traffic, we consider its payoff as zero for all possible actions of the defender. However, running RRM against legitimate users generates wrong information to the legitimate user, and it may cause a problem with troubleshooting. We model this effect by parameter $\delta$. Hence, the payoff of the defender when it runs RRM will be equal to $-\delta - c$.

\input{game}
\input{experiments}

\input{relatedwork}

\input{conclusion}

\vspace{-10pt}
\balance

\endgroup
\bibliographystyle{IEEEtran}
\bibliography{master.bbl}

\vspace{-10pt}
\begin{IEEEbiography}[{\includegraphics[width=1in,height=1.25in,clip,keepaspectratio]{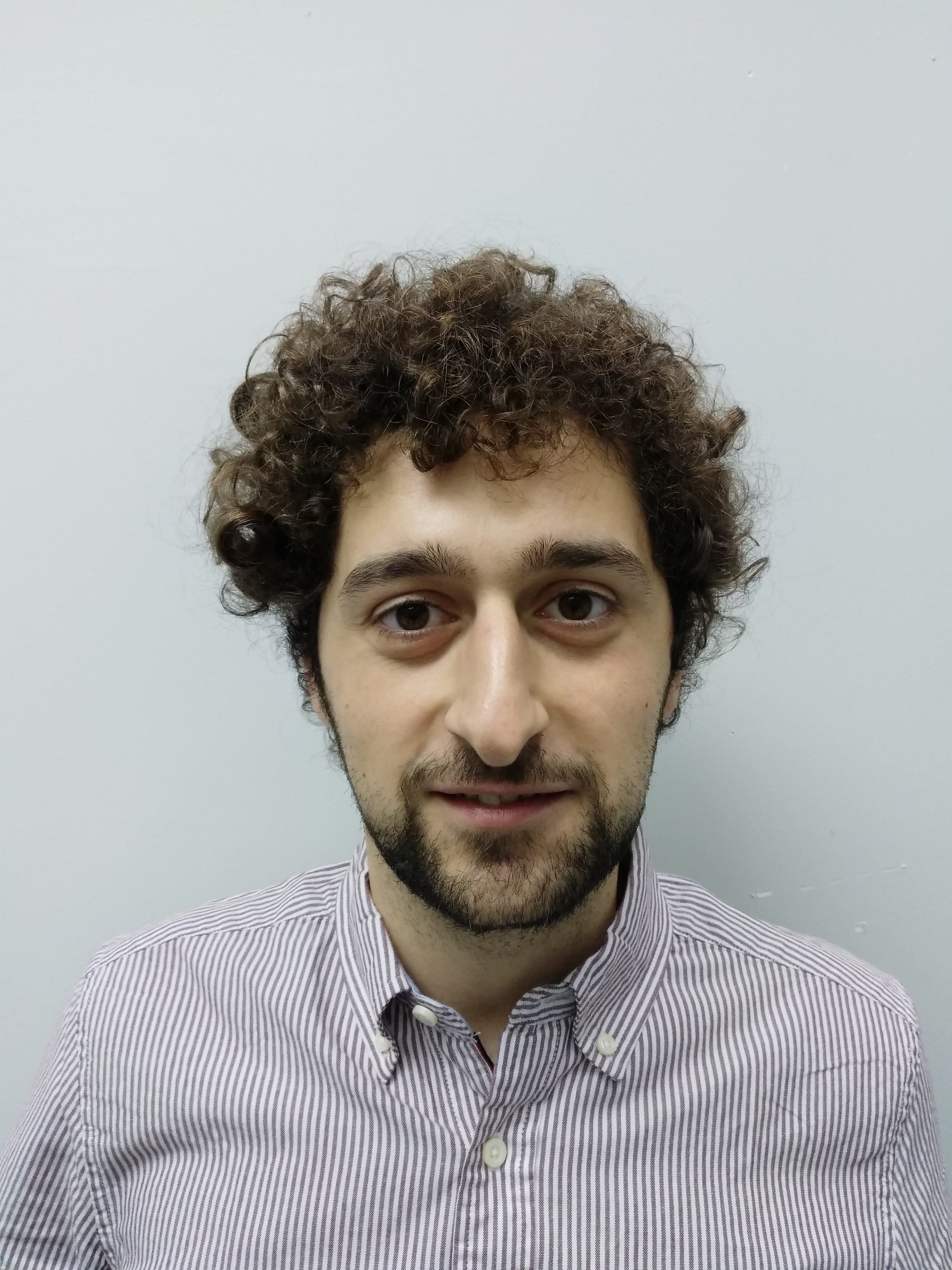}}]{Abdullah Aydeger}
is a Ph.D. candidate in Electrical and Computer Engineering in Florida International University, Miami, USA. He received his M.S. degree from Department of Computer Engineering at Florida International University in 2016 and his B.S. degree in Computer Engineering from Istanbul Technical University in 2013. Prior to his master degree, he did his internship at A*STAR research institute in Singapore. His research interests include SDN, network security, and resiliency.

\end{IEEEbiography}

\begin{IEEEbiography}[{\includegraphics[width=1in,height=1.25in,clip,keepaspectratio]{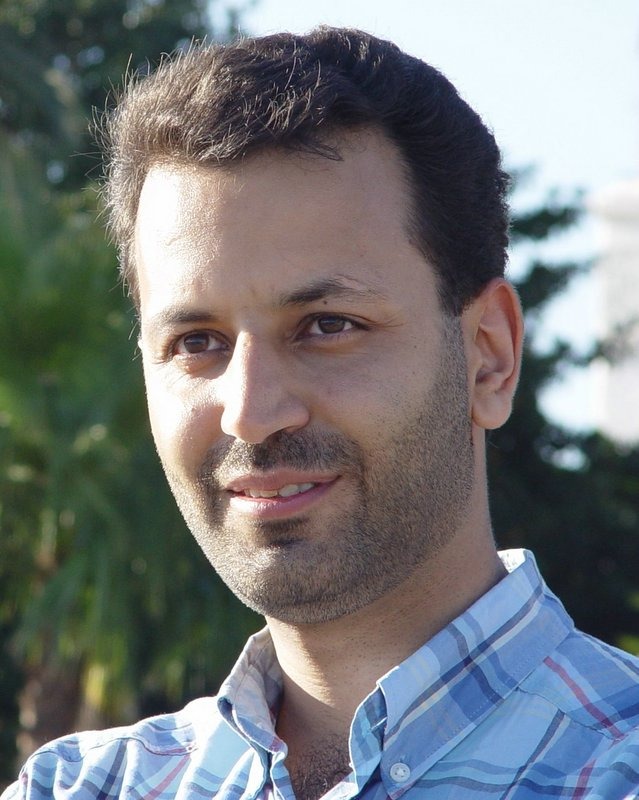}}]{Mohammad Hossein Manshaei} received the B.Sc. degree in electrical engineering and the M.Sc. degree in communication engineering from the Isfahan University of Technology, Iran, in 1997 and 2000, respectively, and the M.Sc. degree in computer science and the Ph.D. degree in computer science and distributed systems from the University of Nice, Sophia-Antipolis, France, in 2002 and 2005, respectively. He did his thesis work at INRIA, Sophia-Antipolis. From 2006 to
2011, he was a Senior Researcher and Lecturer with the Swiss Federal Institute of Technology, Lausanne (EPFL). He held visiting positions at UNCC, NYU, VTech, and UTSA. He is currently a visiting faculty at the Florida International University and an Associate Professor with the Isfahan University of Technology. His research interests include wireless networking, network security and privacy, computational biology, and game theory.
\end{IEEEbiography}

\begin{IEEEbiography}[{\includegraphics[width=1in,height=1.25in,clip,keepaspectratio]{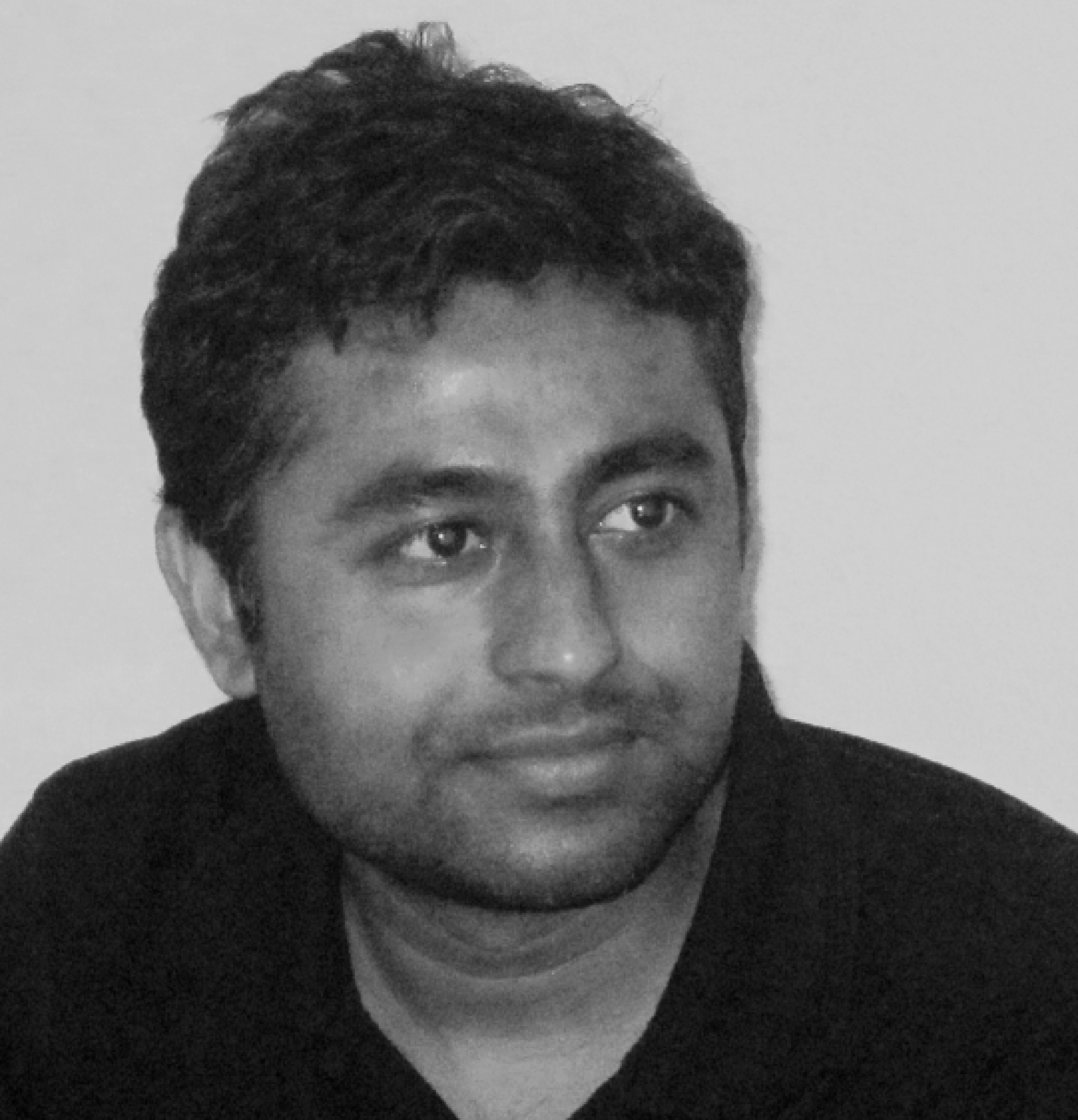}}]{Mohammad Ashiqur Rahman}
is an Assistant Professor in the Department of Electrical and Computer Engineering at Florida International University (FIU), USA. Before joining FIU, he was an Assistant Professor at Tennessee Tech University. He received the BS and MS degrees in computer science and engineering from Bangladesh University of Engineering and Technology (BUET), Dhaka, in 2004 and 2007, respectively, and obtained the PhD degree in computing and information systems from the University of North Carolina at Charlotte in 2015. Rahman's primary research interest covers a wide area of computer networks and communications, within both cyber and cyber-physical systems (CPS). His research focus primarily includes network and information security and CPS dependability. He has already published over 50 peer-reviewed journals and conference papers. He has also served as a member in the technical programs and organization committees for various IEEE and ACM conferences.
\end{IEEEbiography}

\begin{IEEEbiography}[{\includegraphics[width=1in,height=1.25in,clip,keepaspectratio]{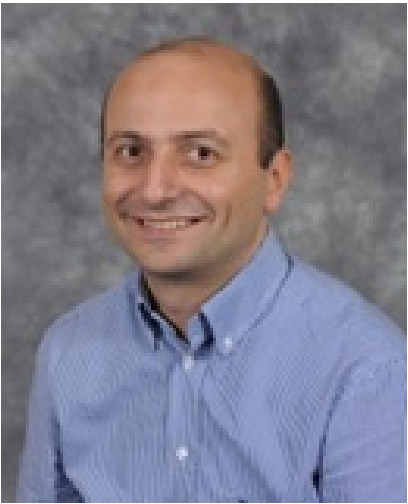}}]{Kemal Akkaya}
is a full professor in the Department of Electrical and Computer Engineering at Florida International University. He received his PhD in Computer Science from University of Maryland Baltimore County in 2005 and joined the department of Computer Science at Southern Illinois University (SIU) as an assistant professor. Dr. Akkaya was an associate professor at SIU from 2011 to 2014. He was also a visiting professor at The George Washington University in Fall 2013. Dr. Akkaya leads the Advanced Wireless and Security Lab (ADWISE) in the ECE Dept. His current research interests include security and privacy, security and privacy, internet-of-things, and cyber-physical systems. Dr. Akkaya is a senior member of IEEE. He is the area editor of Elsevier Ad Hoc Network Journal and serves on the editorial board of IEEE Communication Surveys and Tutorials. He has received ``Top Cited'' article award from Elsevier in 2010.
\end{IEEEbiography}

\end{document}

%% file: introduction.tex
\section{Introduction}

A recent report by Netscout states that Distributed Denial-of-Service (DDoS) attacks will continue to grow \cite{netscout}.
According to the report, last year 1.35 Terabits per second (Tbps) DDoS traffic hit Github, and just after five days of that incident, 1.7 Tbps DDoS traffic launched against an unnamed US-based service provider. 
Even though there are many DDoS defense mechanisms available \cite{zargar2013survey}, they are not capable of competing with some recent types of attacks. For instance, there is a new kind of attack where the attacker does not attack the target network/servers directly \cite{studer2009coremelt,kang2013crossfire}.
This type of attack is known as Stealthy Link Flooding Attack (SLFA). 
The Crossfire attack is an exemplary kind of SLFAs, and it is primarily performed by congesting the communication links surrounding the network of the target servers by sending low-volume traffic over them from many bots in distributed locations. Since the traffic the bots send is legitimate, and they do not attack the servers directly, it is very challenging to detect such attacks using traditional mechanisms. The consequence of this attack is the blockage of external access to the servers while they are still active/running without any problem within the network. 
%


The concept of Moving Target Defense (MTD), in which the defense is done dynamically, often proactively, by introducing agility to the network behavior, is proven useful to defend against such stealthy attacks \cite{venkatesan2016moving, wang2018detecting} \cite{aydeger2016mitigating, aydeger2019moving}. 
This agility brings protection to the system by providing resistance, as it complicates the tasks of an attacker by adding inconsistency or confusion in the knowledge of the system. These features can be implemented in various ways including but not limited to changing IP addresses of network devices, the operating system of servers, and routing information, more often by leveraging the capabilities of Software Defined Networking (SDN).  
MTD, more specifically Random Route Mutation (RRM)~\cite{jafarian2013formal}, is found useful in defending a network system against SLFAs~\cite{aydeger2019moving}. 
However, RRM brings significant overhead to the network since it takes time to update the system characteristic (i.e., routing information), and this process is usually applied not only to malicious users but also to legitimate clients. This is because the defender is not aware of the type of a client. Therefore, the concern of when and for whom to change system parameters in order to minimize the cost of the defender and impact of the attacker becomes a critical issue and it should be investigated thoroughly. 

Even though there are many research works proposing defense mechanisms against SLFAs in the literature, they are mostly reactive solutions in which the attacker made some harm before the attack is mitigated~\cite{xue2018linkscope,gkounis2016interplay}. In addition to these works, there are other kinds of proactive defense mechanisms that apply detection of bots, sending forged responses to the hosts, or similar approaches~\cite{wang2017linkbait,meier2018nethide}. We argue that bot detection might not be possible considering that bots often behave similarly to legitimate hosts, and serving fake replies to hosts is not acceptable since 
that may cause letting the legitimate clients use the longer paths and having more delays.
Therefore, there is a need for proactive, dynamic solutions that can harden the system from the attacks while bringing a minimal service degradation to the hosts. To the best of our knowledge, this work is the first to introduce strategic defense against SLFAs in a proactive manner. 

In order to model such a system that achieves the best efficiency of RRM, we need to design an intelligent defense mechanism that can take into account the above concerns for running RRM. To design such a strategic defense mechanism, in this paper, we apply a signaling game~\cite{noe1988capital} to model and analyze the attack and defense actions together. The results from this analysis assist in applying RRM selectively and appropriately such that the target network can remain sufficiently protected against Crossfire attacks and the legitimate clients experience a minimal cost. In this game, one player constructs a belief about the type of its opponent. This belief is always updated by the opponent's actions. Then, the player can better decide about its optimal strategy given its belief about the type of player from which it receives messages. Essentially, we consider the attacker and the defender as the players of this signaling game. We analyze the game and compute all potential Bayesian Nash equilibria. The game results are then used to decide when and for whom to perform RRM. We name the corresponding defense mechanism as {\it Strategic RRM}. 

%


To evaluate our proposed solution, we first implement a network using Mininet, a virtual SDN testbed~\cite{mininet}, and the defense mechanism on FloodLight Controller \cite{floodlight}. Then, we run experiments with extensive RRM, in which the routing paths are periodically changed, to observe the overhead to the legitimate clients when there is not any attack in place. Next, we run the SLFA and realize the reasonable frequency for periodic RRM, which is used later to compare with {\it Strategic RRM}.
Moreover, we consider two kinds of attackers based on attack approach and capability. 
In each case, we report the number of packet losses when the defender has either no protection, periodic RRM or {\it Strategic RRM}. We show that while periodic RRM provides a significant improvement to the network defense, it introduces packet delay as well as packet losses even when there is no attacker in the system. We also show that {\it Strategic RRM} offers a similar defense performance as periodic RRM while it causes much less overhead.

In summary, our contribution in this paper is threefold: First, we model the attacker and the defender as players of a signaling game and define their actions and payoff models. Second, we solve the game and derive all possible Nash equilibria of this game.
According to the game results, we also design a strategic mechanism to defend against SLFAs. Finally, we evaluate the proposed mechanism by conducting extensive experiments considering different attack approaches and capabilities for different defense strategies. 




This paper is organized as follows. In Section~\ref{Sec:Preliminaries}, we discuss the preliminaries and relevant work. In Section~\ref{sec:game}, we present the game model. In the following section, we analyze the game and design a defense mechanism. Detailed performance evaluation of the proposed work is presented in Section~\ref{sec:eval}. Finally, we conclude the paper in Section~\ref{sec:conclusion}.

%% file: preliminaries.tex
\section{Preliminaries} \label{Sec:Preliminaries}
In this section, we briefly discuss some concepts that we applied in this research.

\subsection{Moving Target Defense}
A system typically offers resources for its clients' usage. A client accesses these resources through interfaces. An interface is often considered as an access points to the system. These access points create different attack surfaces 
with various levels of capabilities and impacts on the system~\cite{manadhata2010attack}. The configurations of these attack surfaces 
usually do not change over time.
The static nature of attack surfaces helps launch attacks easily. 
Once an attacker gains the system information (IP addresses, Operating System, routing, etc.), s/he can use that knowledge to exploit the system vulnerabilities. 

MTD is a defense mechanism that provides security through adding agility/dynamicity to the system, and thus by changing the attack surface~\cite{manadhata2013game}. 
The system's dynamicity deceives the attacker and increases his/her effort to launch a successful attack~\cite{aydeger2019moving}. 
%
There are several methods to perform MTD, including updating software or network characteristics of the system. 
Route mutation-based MTD has been shown useful in defending against SLFAs~\cite{aydeger2016mitigating}.



\subsection{Software Defined Networking}

SDN proposes the separation of routing decisions from forwarding tasks~\cite{mckeown2009software}. 
%
%
An SDN-based network architecture typically considers three layers: application, control, and data. 
While the data layer (i.e., network switch) is responsible for forwarding the data from one port to another, the control layer (i.e., SDN Controller) is responsible for making the decisions on routing rules. The application layer contains network functions including but not limited to Intrusion Detection System (IDS), firewall. It provides great flexibility and efficiency for network management \cite{feamster2014road}. Thus, we leverage SDN to implement our solution. 

\subsection{Attack Model} 
\label{sec:crossfire}
 
Stealthy link flooding attacks, such as the Coremelt \cite{studer2009coremelt} and the Crossfire \cite{kang2013crossfire}, have drawn lots of attention from researchers recently. SLFAs are indirect attacks such that the attacker targets the communication links instead of targeting the server directly~\cite{ramanujan2008protecting}. 


The attack model of this research considers Crossfire attacks as the threat and the Autonomous System (AS) as the target domain. Fig.~\ref{fig:crossfire} shows a Crossfire attack scenario. 
As it can be seen from the Fig.~\ref{fig:crossfire}, there is a target area that the attacker plans to congest using the bots. The attack consists of two phases: reconnaissance and attack. In the first phase, the attacker uses bots residing at distributed locations to obtain how network communication links and routes are formed toward the target area. The attacker uses traceroute packets to obtain the information about the network topology and traffic routes. Leveraging this information, 
s/he obtains the critical communication links that carry most of the data to the target area. 
Finding these links is important as they will be used in the second phase of the attack. 
The attacker should also have the information about other servers in the network. Some of these servers are used to send traffic from bots in order to occupy the communication links. These servers are called decoy servers.
While the bots, often many in number, send regular traffic to the decoy servers to evade Intrusion Detection Systems (IDSs), the critical links suffer with congestion, causing the target area isolated from its clients.

Two different approaches are considered for an attacker: 
\begin{itemize}
    \item \textit{Stealthy Attacker}: In this model, the attacker is not in a rush to conduct the attack. S/he runs reconnaissance, and accordingly creates a target link set. Sometimes s/he stays idle in order to behave like legitimate users. After that, s/he starts the attack phase to exhaust the links' bandwidth.
    
     \item \textit{Aggressive Attacker}: In this approach, the attacker behaves more greedy and transmits more quickly to congest the communication links. Reconnaissance attack is performed, the information is collected and the attack pattern is designed within a given short time for the aggressive attacker.
\end{itemize}

Furthermore, an attacker often has a limited number of bots. Thus, two qualitative capabilities are considered: 
\begin{itemize}
    \item \textit{Decent Attacker}: The attacker has a small number of bots to launch a Crossfire attack.
    
     \item \textit{Strong Attacker}: The attacker has an extended capability (i.e., a significantly increased number of bots) compared to a decent attacker.
     
\end{itemize}


\begin{figure}[tb] 
\centering
\includegraphics[scale=0.17]{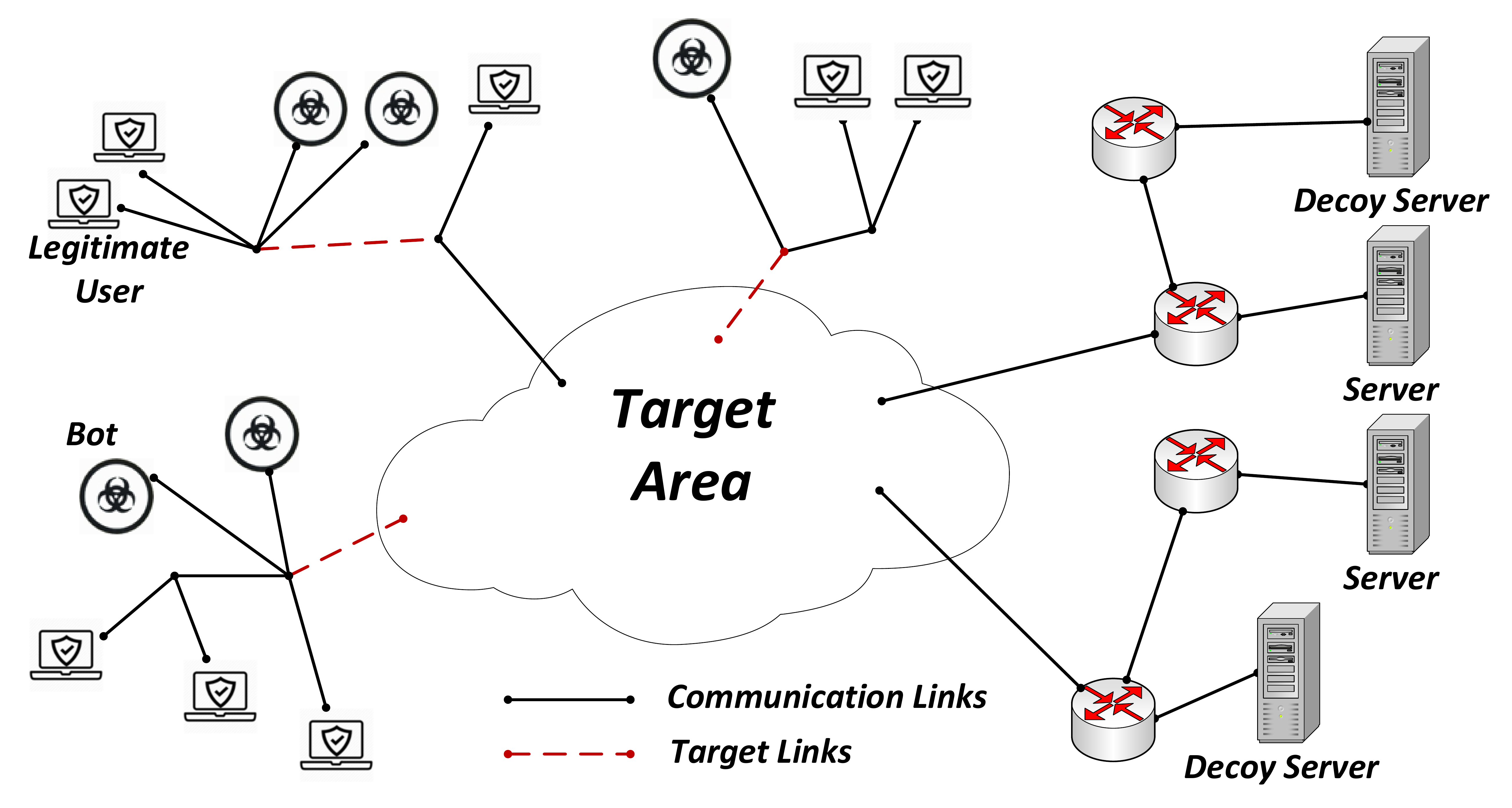}
 \caption{The Crossfire Attack}
\vspace{-1.8em}
\label{fig:crossfire}
\end{figure}

\subsection{Proposed Defense Mechanism} 
\label{sec:mechanism}
Our proposed defense mechanism is based on a \emph{signaling game}. Note that the signaling game is a dynamic game, in which the first move is made by the bot or legitimate user, and the defender is the second player. The game is incomplete information because the defender does not know with whom it is playing. This uncertainty is modeled with Nature. Moreover, we update the knowledge of the defender in repeated interactions with other nodes. Both players decide their actions according to their expected outcome. We explain the details of the game model in Section \ref{sec:game}. The high-level architecture of our defense procedure is shown in Fig. \ref{fig:flowchart}. As can be seen in the Fig. \ref{fig:flowchart}, our system first initializes the parameters depending on the network environment and topology. Then, it checks if the conditions to trigger defense strategy is satisfied. If the decision is made as to defend, RRM is employed, and system parameters are updated. If not, the defender stays idle until it receives an action from the attacker. Each time an action received from the attacker, the defender updates his/her parameters and checks if it is time to defend. 

\begin{figure}[tb] 
\centering
\includegraphics[scale=0.4]{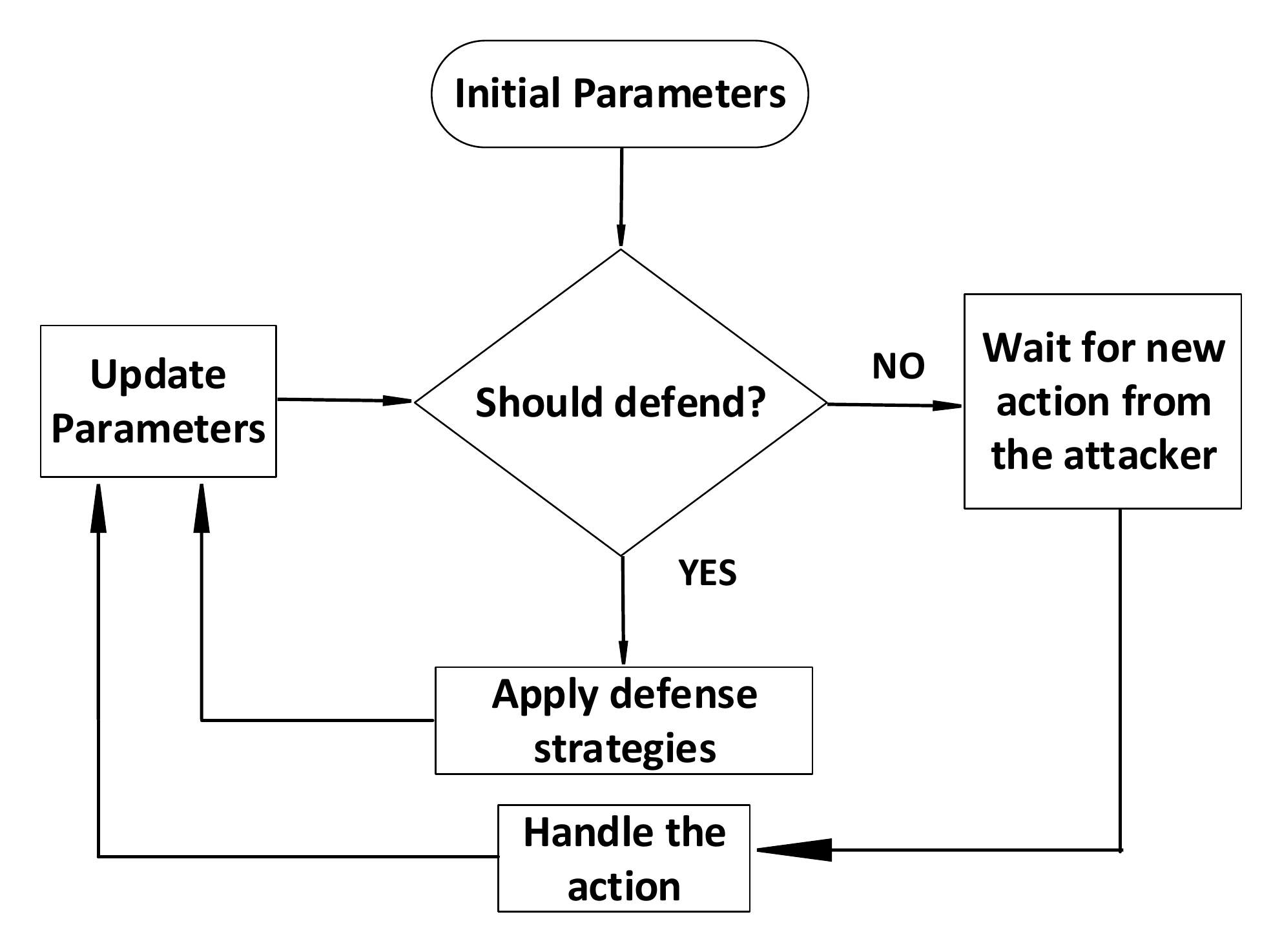}
 \caption{The procedure of defense decision}
\vspace{-0.5em}
\label{fig:flowchart}
\end{figure}

%% file: game.tex
\section{Game Analysis and Protocol Design}
\label{Sec:Result}
	
In the following, we first examine game $\mathcal{G}^{cf}$ for the existence and properties of pure strategy Perfect Bayesian Nash Equilibria (PBNE). We then use our analysis to design a defensive protocol to optimize defender strategies against Crossfire attacks.

\subsection{Game Analysis: Perfect Bayesian Nash Equilibrium}
\label{sub-GameAnalusis}

In complete information or non-Bayesian games, a strategy profile is a Nash equilibrium (NE) if every strategy in that profile is a best response to every other strategy. However, players in Bayesian games would like to maximize their expected payoffs, given their beliefs about the other players~\cite{shoham2008multiagent}. A PBNE is characterized as a strategy profile and belief that satisfy the following four requirements~\cite{gibbons1992primer}:


	\vspace{3pt}
	\noindent{\textbf{Requirement 1:}} After observing  any message $m_{j}$ from sender $j$, the defender must have a belief about which types could have sent $m_{j}$. Denote this belief by the  probability distribution $\mu(t_{i}|m_{j})$, where $\mu(t_{i}|m_{j})\ge 0 $ for each type $t_{i}$, and $\sum_{t_{i}\in T}{\mu(t_{i}|m_{j})}=1$ 
	
	
	\vspace{3pt}
	\noindent{\textbf{Requirement 2:}} For each message $m_{j}$, the defender's action $a^*(m_{j})$ must maximize his expected utility $u_d$, given the belief $\mu(t_{i}|m_{j})$ about which type could have sent $m_{j}$. That is, $a^*(m_{j})$ satisfies:
	\begin{equation*}
		\max_{m_j \in M}{\sum_{t_{i}\in T}{\mu(t_{i}|m_{j})u_{d}(t_{i},m_{j},a(m_j))}}
	\end{equation*}
	
	\vspace{3pt}
	\noindent{\textbf{Requirement 3:}}
	For each type $t_{i}$, the sender's (whether a bot or a legitimate user) message $m^{*}(t_{i})$ must maximize his utility ($u_d$), given the defender's strategy $a^{*}(m_{j})$. That is, $m^*(t_{i})$ satisfies:
	\begin{equation*}
		\max_{m_{j}\in M}{u_{d}(t_{i},m_{j},a^{*}(m_{j}))}
	\end{equation*}
	
	\vspace{3pt}
	\noindent{\textbf{Requirement 4:}} For each $m_{j}\in M$, if there exists type $t_{i}$ such that $m^{*}(t_{i})=m_{j}$, then the defender's belief at the information set corresponding to $m_{j}$ must follow from Bayes' rule and the sender's strategy:
	\begin{equation*}
		\mu(t_{i}|m_{j})=\frac{p(t_{i})}{\sum_{t_{i}\in T_{j}}{p(t_{i})}}
	\end{equation*}
	where $T_{j}$ denotes the set of types that send the message  $m_{j}$. Considering the above requirements, we can now define the PBNE.

	\begin{definition}
		A pure-strategy PBNE in a signaling game  is a pair  of strategy $m^{*}(t_{i})$ and $a^{*}(m_{j})$  and a belief $\mu(t_{i}|m_{j})$ satisfying Requirements 1 to 4.
	\end{definition}

	In the following, we use $(p,1-p)$ and $(q,1-q)$ to denote the second player's (the defender) beliefs at its two information sets. Recall that for the defined signaling game in Figure~\ref{fig:signaling-game}, the sender's pure strategy determined by an ordered pair $(m(t_{1}),m(t_{2}))$ where $m(t_{1})$ and $m(t_{2})$ are the chosen strategies by user types $t_{1}$ and $t_{2}$, respectively. Note that in our model $t_1$ and $t_2$ are bot and legitimate types. Similarly, the defender's strategy is determined by an ordered pair  $(a(N), a(G))$, in which $a(N)$ and $a(G)$ demonstrate the defender strategy following the sender's reconnaissance and regular traffic signals, respectively. 
	
	Furthermore,  a pure strategy  PBNE profile  is   determined as tuple  
	$\{\mathcal{S}_{1},\mathcal{S}_{2},p,q\}$, in which $\mathcal{S}_{1}$  is the pair of the sender strategy chosen by each type (whether bot or legitimate user),  $\mathcal{S}_{2}$ is the pair of defender strategy in response to each signal, and $p$ and $q$ are attacker belief concerning the type of sender for reconnaissance ($N$) or regular ($G$) signal, respectively. According to the sender pure strategy, two kinds of PBNE could exist in signaling game, called \emph{pooling} and \emph{separating}.
	

	A PBNE is called \textit{pooling  equilibrium} if $m(t_{1})=m(t_{2})$. In other words,  the bot and legitimate user send the same signal, regardless of their types. In contrast, a PBNE is called \textit{separating  equilibrium} if $m(t_{1})\neq m(t_{2})$, i.e., the bot and legitimate users send a different signal, depending on their types. We now examine $\mathcal{G}^{cf}$ for (pure) PBNE. We first probe the existence of pooling equilibria. 

Given the definition of pooling and separating equilibrium, in the following, we derive all conditions under which there exist these Nash equilibrium profiles in our defined game. In other words, since the values of payoffs vary given the topology of the network and the period of sending different packets and probes in the reconnaissance phase, there would be different conditions to be checked for the existence of these Nash equilibrium points.

	\begin{theorem}
		\label{Thm:PoolingN}
		For any values of  $\theta$, there exists a pooling equilibrium on $N$, in $\mathcal{G}^{cf}$ signaling game.
	\end{theorem}

\begin{proof} Let us suppose that there exists a pooling NE with $(N,N)$ strategy for the sender.  Then the defender's information set corresponding to $N$ is on the equilibrium path, so the defender's beliefs $(p,1-p)$ at this information set is determined by Bayes' rule and sender's strategy: $p=\theta$. We first compute the expected payoff of the defender given its belief. The defender's expected payoff for playing $R$ is:
		\begin{equation}
			\label{Exp_RRM}
			\theta \times(\lambda-\frac{1}{f}\alpha-c) +(1-\theta )\times(-\delta-c)
		\end{equation}

And defender's expected payoff for playing $\overline{R}$ is:
		\begin{equation}
			\label{Exp_NotRRM}
			\theta \times(-\alpha)+(1-\theta )\times(0)
		\end{equation}

Comparing the above payoffs for the defender we can define a threshold for belief, called $\theta^{*}$, which is equal to $\frac{\delta +c}{ \delta +\alpha (1+\lambda-\frac{1}{f}) }$. We first assume that $\theta^{*} \le 0$, then the following two cases could take place for the dominant strategy of the defender:

\begin{itemize}
\item $\theta \ge \theta^{*}$ := $\frac{\delta +c}{ \delta +\alpha (1+\lambda-\frac{1}{f})}$: Therefore playing $R$ dominates $\overline{R}$, following  $N$ signal, which can be easily verified by comparing the expected payoffs presented in Equations (\ref{Exp_RRM}) and (\ref{Exp_NotRRM}). 

Now we should check whether the senders have incentives to deviate from the $N$ strategies, given the defender strategy which is $R$ in this case. If the defender chooses strategy $R$ for responding to message $G$, there is no incentive for the senders to deviate from their strategies. In fact, in that case, the sender of type 1 ($t_{1}$ or Bot) achieves $-\beta$ instead of $\frac{1}{f} \alpha - \beta$. The sender of type 2 ($t_{2}$ or legitimate user)  obtains $0$ in both cases.  Hence, there is no incentive for the senders to deviate from strategy $N$. 

It remains to consider the defender's belief at the information set corresponding to $G$ (i.e., off the equilibrium path). We need to show if the strategy of playing $R$ is optimal given this belief. For this purpose and given that the $R$ is the best response when $\theta \ge  \theta^{*}$, we should calculate the expected payoffs of the defender when it plays $R$ and $\overline{R}$ following strategy $G$. These payoffs are $q  \times -c + (1-q) \times (-\delta -c )$ and $q  \times 0 + (1-q)  \times 0 $. Since there are no values for $q$ that makes the payoff of defender greater for playing $R$, there is no pooling on $(G,G)$ when the defender plays $(R,R)$. 

Similarly, to verify if there exists a NE where the defender plays $(R,\overline{R})$ we should first show that there are no incentives for the sender to deviate from the pooling $(N,N)$ strategy. This time considering the defender strategy $(R,\overline{R})$, both types of senders do not have any incentive to deviate from $N$ and play $G$ as their payoff would be decreased from $\frac{1}{f}\alpha - \beta$ to $-\beta$ for the bot player and the legitimate user, its payoff remains $0$. Now, we should again calculate the payoff of the defender given its belief $q$. In other words, this time we need to compute the values of $q$, where   $q  \times -c + (1-q) \times (-\delta -c ) \leq 0$. Hence for all values of $q \le 0$, the payoff of the defender would be greater when it plays $\overline{R}$. Then, we can conclude that there exists one pooling equilibrium $\{ (N,N),(R,\overline{R}),p=\theta,q\} $ for any $q$  in $\mathbb{G}_{cf}$ when $\theta \ge \theta^{*}$.

\item $\theta \le \theta^{*}$ := $\frac{\delta +c}{ \delta +\alpha (1+\lambda-\frac{1}{f})}$: In this case the best response of the defender following $N$ signal is $\overline{R}$. Similar to the previous case, we should first verify if there is any incentive for the senders to deviate from $N$, if the defender plays $\overline{R}$ and $\overline{R}$ off the equilibrium path. Since both types will not gain any extra benefits (for bots the payoff decreases from $\alpha - \beta$ to $\beta$ and for the legitimate users there is no differences), there is no incentives for them to deviate from playing $N$. Similar calculations can be done for the payoff of the defender, off the equilibrium path to find the possible values for $q$. Similar to the previous case, there are no values for $q$, where the expected payoff of playing $R$ would be bigger than $\overline{R}$. Considering all possible deviations for the case that the defender plays $\overline{R}$ and the values of the belief for $q$ there exists another pooling equilibrium, where $\{ (N,N),(\overline{R},\overline{R}),p=\theta,q\} $ for any $q$ in $\mathcal{G}^{cf}$ when $\theta \le \theta^{*}$.
\end{itemize}
\end{proof}	
	
	
	Theorem~\ref{Thm:PoolingN} represents that if the selected strategy of both types of the senders are $N$, there is an equilibrium considering the belief of the defender for the possible attacker. In this case, depending on the value of $\theta$, the defender should select one of the strategies $R$ or $ \overline{R}$ upon receiving signal $N$. In other words, if the probability of the sender being a bot (i.e., $ \theta $) is greater than $\theta^{*} $, the defender should run $RRM$. Otherwise, the best response for the defender is playing $\overline{R}$ strategy. In the case of $\theta \ge \theta^{*}$ and $\theta \le \theta^{*}$, if users send reconnaissance packets, the defender's response  will be $R$  and $\overline{R}$, respectively.

	\begin{table*}[t]
		\centering
		\caption{Equilibria and Their Conditions}
		\label{Table_PBNE}
		\begin{tabular}{|c|c|c|c|c|c|c|}
			\hline
			\multirow{2}{*}{Theorem}&
			\multirow{2}{*}{Conditions}&\multirow{2}{*}{Range of $\theta$} & \multirow{2}{*}{$\#$}& \multirow{2}{*}{$PBNE$}& \multicolumn{2}{|c|}{Condition on Beliefs}\\
			\cline{6-7}
			&ٍ & && &On-equilibrium & Off-equilibrium\\
			\hline
			Theorem 1& $-$&$\theta\ge\theta^{*}$& $\mathcal{PBNE}_{1}$ & $\{(N,N),(R,\overline{R}),p,q\}$ &$p=\theta$ & $ \forall q$\\
			\hline
			Theorem 1& $-$
			&$\theta\le\theta^{*}$ &$\mathcal{PBNE}_{2}$ & $\{(N,N),(\overline{R},\overline{R},),p,q\}$ &$p=\theta$ & $ \forall q$\\
			\hline
			Theorem 4&$\lambda \ge\lambda^* $&$\forall \theta^{*}$&$\mathcal{PBNE}_{3}$&$\{(N,G),(R,\overline{R}),p,q\}$&$p=1$&$q=0$\\
			\hline 
		Theorem 4&$\lambda \le\lambda^* $&$\forall \theta^{*}$&$\mathcal{PBNE}_{4}$&$\{(N,G),(\overline{R},\overline{R}),p,q\}$&$p=1$&$q=0$\\
			\hline 
			
		\end{tabular}
	\end{table*}

	\begin{theorem}
\label{Thm:PoolingG}
		For any values of  $\theta$, there does not exist any pooling equilibrium on $G$, in $\mathcal{G}^{cf}$ signaling game.
	\end{theorem}

\begin{proof}
In the game $\mathcal{G}^{cf}$ for any values of belief $\theta$, the defender's best respond to pooling  strategy of $(G,G)$ is always $\overline{R}$. In other words, the defender does not perform RRM in this case. Consequently, we should see if the senders have any incentive to deviate from $G$ strategy. Let us consider two possible strategies of the defender off the equilibrium path (when it believes in playing with Bot with probability $p$). Since by deviating from $G$ to $N$, the sender of type $t_1$ (Bot) can always increase his/her payoff from $-\beta$ to $\frac{1}{f}\alpha - \beta$ (when the defender plays $R$ off the equilibrium path) or from $-\beta$ to $\alpha - \beta$ (when the defender plays $\overline{R}$ off the equilibrium path), $(G,G)$ cannot be at any pooling equilibrium.
\end{proof}

	\begin{theorem}
\label{Thm:Separating1}
		There is no separating equilibrium on $(G,N)$ in the $\mathcal{G}^{cf}$ signaling game.
	\end{theorem}

\begin{proof}
Suppose $(G,N)$ is a pair of senders' strategy, then both of the defender's information sets are on the equilibrium path, so both beliefs are determined by Bayes' rule and sender strategy: $q=1$ and $p=0$. Defender's best response following  these beliefs is always $\overline{R}$, for both types of the sender. Hence, we should check if the sender's strategy is optimal given the  defender strategy. If the sender of type 1 deviates by playing $N$ signal instead of $G$, the defender responds with $\overline{R}$, giving $t_{1}$ (i.e., the Bot) a payoff of $\alpha - \beta$, which exceeds $t_{1}$'s payoff of $-\beta$ from playing $G$. Thus, $(G,N)$ cannot establish any separating equilibrium.
\end{proof}

	\begin{theorem}
\label{Thm:Separating2}
		There are two classes of separating equilibrium $\{ (N,G),(\overline{R},\overline{R}),p=1,q=0\} $ if $\lambda \le \alpha (1/f-1)+c $ and $\{ (N,G),(R,\overline{R}),p=1,q=0\} $ $\lambda \ge \alpha (1/f-1)+c $  in the $\mathcal{G}^{cf}$ signaling game.
	\end{theorem}

\begin{proof}
Similar to Theorem~\ref{Thm:Separating1}, we first assume that $(N,G)$ is a pair of senders' strategy. Then both of the defender's information sets are on the equilibrium path: $p=1$ and $q=0$. Considering $(N,G)$ strategy of the senders, the defender's best response following these beliefs can be calculated by comparing the payoffs of the defenders. For sender with type $t_1$ (i.e., bot), the best response of the defender is calculated by comparing the following two payoffs: $\lambda - \frac{1}{f} \alpha -c$ and $-\alpha$. Two cases can be identified given that $\lambda^* = \alpha (\frac{1}{f} -1 )+ c$:
\begin{itemize}
\item $\lambda \ge\lambda^* $: The best response to play $N$ by type $t_1$ is $R$ and the best response to type $G$ by type $t_2$ is $\overline{R}$. We should check, if the sender's strategy is optimal given the  defender strategy. Since the bot payoff would be decreased from $\frac{1}{f} \alpha - \beta$ to $-\beta$ and there is no difference for the legitimate payoff, we can conclude that the $\{ (N,G),(\overline{R},\overline{R}),p=1,q=0\} $ is a PBNE. We name it as $\mathcal{PBNE}_{3}$ for later references. 
\item $\lambda \le \lambda^* $: In this case, the best response to play $N$ by type $t_1$ is $\overline{R}$ and the best response to type $G$ by type $t_2$ is $\overline{R}$. Similar to the previous case there is no incentive for the senders to deviate from $(N,G)$ strategy.
\end{itemize}
\end{proof}

The above results present all possible PBNEs of game $\mathcal{G}^{cf}$. Considering these equilibria, we can now provide the best plan of actions for the defender given its belief about the sender's type and its payoff at different strategies.

\subsection{Protocol Design}
\label{sub-protocol}

In this section, we design a strategic Crossfire defense mechanism to optimize the strategy of the defender. Table \ref{Table_PBNE} summarizes the results presented in Theorems~\ref{Thm:PoolingN}, \ref{Thm:PoolingG}, \ref{Thm:Separating1}, and \ref{Thm:Separating2}. All possible separating and pooling PBNEs are displayed. Leveraging these results, we design our proposed strategic Crossfire defense mechanism, namely \emph{Strategic RRM}. The pseudocode of this proposed mechanism is given in Algorithm \ref{Algo_Strategic}. 

In the Algorithm \ref{Algo_Strategic}, we first pick a value for the frequency ($f$) of RRM that represents how often the defender can perform RRM. It is decided through our preliminary experiments. 
The minimum period of running RRM is shown as $P$ which is assigned to $1/f$. 
Then the belief for each client is initialized to zero at the beginning of the experiments. 
The defender always updates its belief each time it receives a packet, and it computes the potential gain/payoff.  
Later, if the belief is greater than $\theta^*$ and the received packet is reconnaissance, the defender chooses to run $RRM$ at the end of $P$ period. In addition to that, if the value of $\lambda$ is greater than $\lambda^*$, the defender should also run the RRM. Hence, the optimal decision for the defender is obtained according to equilibria $\mathcal{PBNE}_{1}$ and $\mathcal{PBNE}_{3}$. 
The same processes after the initialization will be performed every $P$ seconds. 



\begin{algorithm}[t]
\caption{Strategic Crossfire Defense Mechanism}
\label{Algo_Strategic}
\small
\begin{algorithmic}[1]
\STATE $h$: Host ID (i.e., its IP address) communicating with the server
\STATE $t^h_{alive}$: The time that host $h$ is still transmitting packets
\STATE $f$: The frequency at which we can run RRM
\STATE $P$ $\Leftarrow$  $\frac{1}{f}$
\STATE $n := 1$
\STATE $\theta_h$ $\Leftarrow$   Initialize the belief for $h$ 

\WHILE{$((n-1) \times T) \le t^h_{alive}$}
	\STATE Estimate/calculate $\alpha$, $\beta$, $\lambda$, $\delta$, and $c$.

	\FOR {Each packet received from $h$ between $(n - 1)T$ and $nT$}
		\STATE  Update $\theta_h$ according to the packet type (signal)
		\STATE Compute $\theta^*$ (Theorem~\ref{Thm:PoolingN})
		\STATE Compute $\lambda^*$  (Theorem~\ref{Thm:Separating2})
	\ENDFOR

	\IF {$\lambda \ge \lambda^*$}
		\STATE Perform RRM	
	\ELSIF {$\theta \ge \theta^*$} 
		\STATE Perform RRM	
	\ENDIF	
	
	\STATE $n := n + 1$
\ENDWHILE
\end{algorithmic}
\normalsize
\end{algorithm}

%% file: experiments.tex
\section{Evaluation}
\label{sec:eval}


In this section, we briefly discuss our experimental setup and evaluation metrics. Then, we present the findings from the experiments. 

\subsection{Experiment Setup}

The network environment is implemented using Mininet \cite{mininet} and FloodLight \cite{floodlight} is used as an SDN Controller. The proposed signaling game-based defense strategy is developed as an application on top of the FloodLight Controller.
In the experimental setup, we consider the network topology with 20 switches, 3 decoy servers, 1 target server, and several clients, as shown in Fig. \ref{fig:topology}. The number of clients differs for each experiment based on the attack model and the purpose of the experiment. The target server is encircled by red and decoy servers are placed interior part of the network topology in Fig. \ref{fig:topology}. The bandwidth of each link in the network is configured as 100 Mbps. 
A python-based client/server application is implemented for the communication between host-target servers. 
We implement different attacker models and apply in experiments to evaluate them. The configuration of the network is further explained below.

\begin{figure}[tb] 
\centering
\includegraphics[width=\columnwidth]{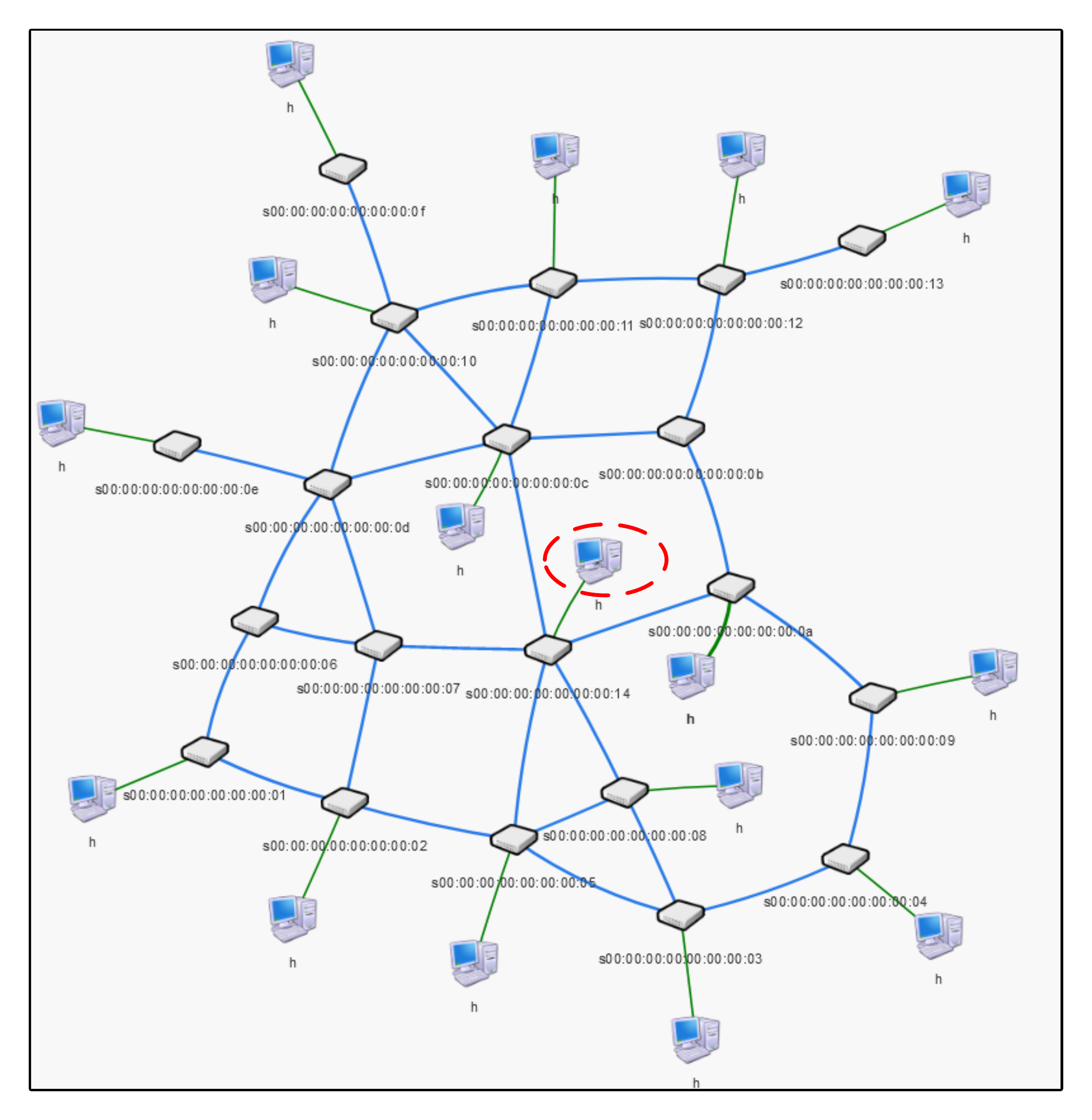}
\vspace{-1.8em}
 \caption{Experimental Network Topology}
\vspace{-1.8em}
\label{fig:topology}
\end{figure}



\subsubsection{Defender's Setup} \label{sec:defender}


In our experiments, we consider 3 different configurations for the defender. We name them as \textit{No Defense}, \textit{Periodic RRM}, and \textit{Strategic RRM}. 
\begin{itemize}
    \item \textit{No Defense}: To observe the worst-case attack scenario the network environment is implemented without any defense mechanism. In this case, the attacker 
    can easily gain the (static) routing information by doing reconnaissance attacks. Next, s/he attacks the target links which are found by using that static information.

    \item \textit{Periodic RRM}: The second defense setup is based on performing RRM periodically. The routes are changed periodically to the alternative ones that are pre-calculated. We choose different values of frequencies to run RRM in order to observe the impact.
    
    \item \textit{Strategic RRM}: It is our proposed signaling game-based defense mechanism. In this case, the defender takes actions based on the attacker's actions as explained in Section \ref{sec:strategymodel}. The mechanism is given in Algorithm \ref{Algo_Strategic} and is implemented accordingly. 
    
\end{itemize}

\subsubsection{Attacker's Setup} \label{sec:attacker}

In a Crossfire attack setup, a decoy server does not receive a large amount of traffic from one or multiple bots simultaneously so that the activity is not considered as suspicious at the network administrator's side. 
Thus, we limit each bot's communication to only 1 decoy server with a regular amount of traffic.
Each bot sends/receives approximately 5 Mbit of data per second in order to behave as legitimate. 


We implement two different attacker behaviors, namely \textit{aggressive} and \textit{stealthy}, with two different capabilities, such as \textit{decent} and \textit{strong}, as mentioned in Section~\ref{sec:crossfire}. 
The aggressive attacker's time for reconnaissance is fixed to 1 minute in experiments. The decent attacker has a limited number of bots, which is as much as the number of legitimate clients in the network. Meanwhile, the strong attacker is equipped with a doubled number of bots. Even though the attack impact is expected to increase, it should also be noted that the more bots the attacker employs, the  higher cost is incurred on him/her. We run our experiments on each case to show how effective our solution would be 
with different attack models.

\begin {table*}[ht]
\begin{center}
\caption {Delay Caused by RRM (in microsecond)} \label{tab:delayrrm} 
  \begin{tabular}{ | l |  c | c |  c |  c |  c | c |  c | }
    \hline
     \textbf{\textit{Client}} & \textbf{\textit{Delay }} & 
     \multicolumn{6}{|c|}{\textbf{\textit{Delay with RRM}}}  \\
     \cline{3-8}
     \textbf{\textit{ID}} & \textbf{\textit{Without  }} & 
     \multicolumn{3}{|c|}{\textbf{\textit{Optimum size Path}}} & 
     \multicolumn{3}{|c|}{\textbf{\textit{Any size Path}}}  \\
     \cline{3-8}
     \textbf{\textit{}} & \textbf{\textit{RRM}} &\textit{\textbf{60-seconds RRM}} & \textit{\textbf{30-seconds RRM}} & \textit{\textbf{10-seconds RRM}} & \textit{\textbf{60-seconds RRM}} & \textit{\textbf{30-seconds RRM}} & \textit{\textbf{10-seconds RRM}}\\
     \cline{1-2}
     \hline 
    1  & 2853 & 2892 & 3040 & 3398 & 3781 & 3717 & 4053 \\ \hline 
    2  & 2854 & 2898 & 3050 & 3406 & 3761 & 3616 & 4145 \\ \hline 
    3  & 1497 & 1533 & 1629 & 1882 & 1973 & 2075 & 2351  \\ \hline 
    4  & 1499 & 2161 & 2322 & 2698 & 2986 & 3205 & 3578  \\ \hline 
  \end{tabular}
\end{center}
\vspace{-2.5em}
\end{table*}

\vspace{-10pt}
\subsection{Evaluation Metrics}



The application of RRM introduces some overhead cost on the defender since it creates extra packets in the network. 
This overhead can cause Quality of Service (QoS) problems in a form of increased communication latency/packet delay or a higher number of packet losses for the legitimate users. 
We specifically consider the following 
metrics to be measured in our experiments:

\begin{itemize}
    \item \textit{Delay for using longer path:} Using a randomized alternative path can cause longer end-to-end delay since the route may not be the optimal/shortest anymore. This metric shows the increase in delay compared to the optimal path. 
    
    \item \textit{Delay for flow table updates:} Flow tables are updated whenever RRM is triggered. This process of updating flow table entries of switches takes time. This time is measured to show additional delays that legitimate clients are exposed to. 
    Related to this issue, some of the packets may drop if the queue of a switch becomes full and it cannot handle any more packets whenever the flow table is updated even though no attack is being taken place. 
    
    \item \textit{The number of packet losses:} 
    This metric represents the number of packets lost for legitimate users due to attacks or the execution of RRM. 
    
    
\end{itemize}

\vspace{-10pt}
\subsection{Experimental Results}

We run experiments on different network configurations to observe each evaluation metric separately. We first show the overhead of performing RRM 
on the network (i.e., legitimate users) varying the RRM frequency, i.e., the interval period between two subsequent route mutations.
We find the optimal period among them, and use this \textit{Periodic RRM} to compare  with the proposed \textit{Strategic RRM} by running experiments with different adversary models. 

\subsubsection{Observing RRM Costs without any Attacker}

We run our experiments without considering any malicious activities in the network in order to measure the defender's cost (i.e., the cost imposed on its clients). In the experiments, we measure the additional delay that is caused due to updating flow tables as well as using longer paths when RRM is applied. %
We first change the RRM frequency to observe the flow table update cost. Here, the routing paths are kept the same. 
Then, to assess the delay for using longer paths, we run experiments selecting optimal (shortest length) paths as well as non-optimal (alternative) paths. 

Table~\ref{tab:delayrrm} presents the average packet delay for 4 randomly selected legitimate clients. 
We can easily see that a higher frequency of RRM (i.e., when RRM interval period is 10 seconds) brings additional delay to the clients if alternative path sizes are the same as the optimum path. The slight difference between each column is caused by flow table updating. As an example, we can take a look at the delays for Client 1: during the 60-seconds period RRM, the average delay is 2892 milliseconds, while the delay becomes 3040 milliseconds and 3398 milliseconds when the periods are  30-seconds and 10-seconds, respectively. Even though the difference between them looks negligible (in milliseconds), it should be noted that the average delay is computed for approximately 100 thousands of packets in a 10-minute long experiment. Thus, the sum of additional delays that are caused by flow table updates, is a significant cost considering the delay increase for 100 thousands of packets where only a few times flow table updates are performed. 
%
We can also observe in Table \ref{tab:delayrrm} that if we use 
alternative (i.e., non-optimal) paths for clients, the average delay increases significantly. This is because the increased number of hops adds further delay to the packet communication. 

Next, we assess if the route updates can cause packet drops. 
The experimental setup remains the same as the previous ones except
we consider a higher number of clients to send more simultaneous packets in order to occupy the queues/buffers at the switches and simulate an environment that would more likely occur during Crossfire attacks. We use small sized packets (100 bytes each) so that the link bandwidth cannot be a cause of packet dropping. The results from the experiments are presented in Table \ref{tab:switchingcost}. We observe that a higher frequency of RRM causes a larger number of packet drops. This is reasonable considering switches are busier updating their flow tables in the cases where more often route changes are occurred. In other words, some of the packets cannot be handled on time and are dropped because of overflowing of the queue. 

\begin {table}[t]
\begin{center}
\caption {Packet Drop Caused by Flow Table Update}
\vspace{-0.5em} \label{tab:switchingcost} 
  \begin{tabular}{ | c |  c | c |}
    \hline
      10-seconds RRM & 30-seconds RRM & 60-seconds RRM \\ \hline
    10.1\%  &   8.6 \%  & 1.6\%  \\ \hline 
  \end{tabular}
\end{center}
\vspace{-2.5em}
\end {table}

The 60-seconds RRM case causes fewer delays and less packet loss compared to 30-seconds and 10-seconds RRM cases if there is no attack, as shown in Tables~\ref{tab:delayrrm} and~\ref{tab:switchingcost}. 
In other words, the 10-seconds RRM case causes lots of overhead compared to 60-seconds RRM. Thus, we opt to run 60-seconds RRM in the following experiments to compare with \textit{Strategic RRM}. 

\begin {table*}[t]
\begin{center}
\caption {Cost Metric for Different Attacker Model} \label{tab:defendercost} 
  \begin{tabular}{ | l |  c |  c | c | c | c | c |}
    \hline
     Attacker   &  \multicolumn{3}{|c|}{Packet Loss} & \multicolumn{3}{|c|}{Average Delay (in milliseconds) }\\
     \cline{2-7}
      Model &  No Defense & Periodic RRM & Strategic RRM & No Defense & Periodic RRM & Strategic RRM \\ \hline
    Strong-Aggressive  &   31\% & 22\% & 16\% & 1172 & 989 & 438  \\ \hline 
    Strong-Stealthy   &  15\%  & 7\%  & 10\% & 579 & 380 & 278 \\ \hline 
    Decent-Aggressive   &   7\% & 2\% & 2\%  & 464   & 251 & 200 \\ \hline 
    Decent-Stealthy  &  3\%  & 0.7\% & 1\% & 263 & 168 & 141 \\ \hline 
  \end{tabular}
\end{center}
\vspace{-1.5em}
\end {table*}



\vspace{-0.5em}
\begin{figure*}[htb]
\setlength{\abovecaptionskip}{0.2cm}
\setlength{\belowcaptionskip}{0cm}
\centering 
\subfloat[Aggressive Behavior Model \label{fig:powerful-agg}]{\includegraphics[width=3.5in,height =2in]{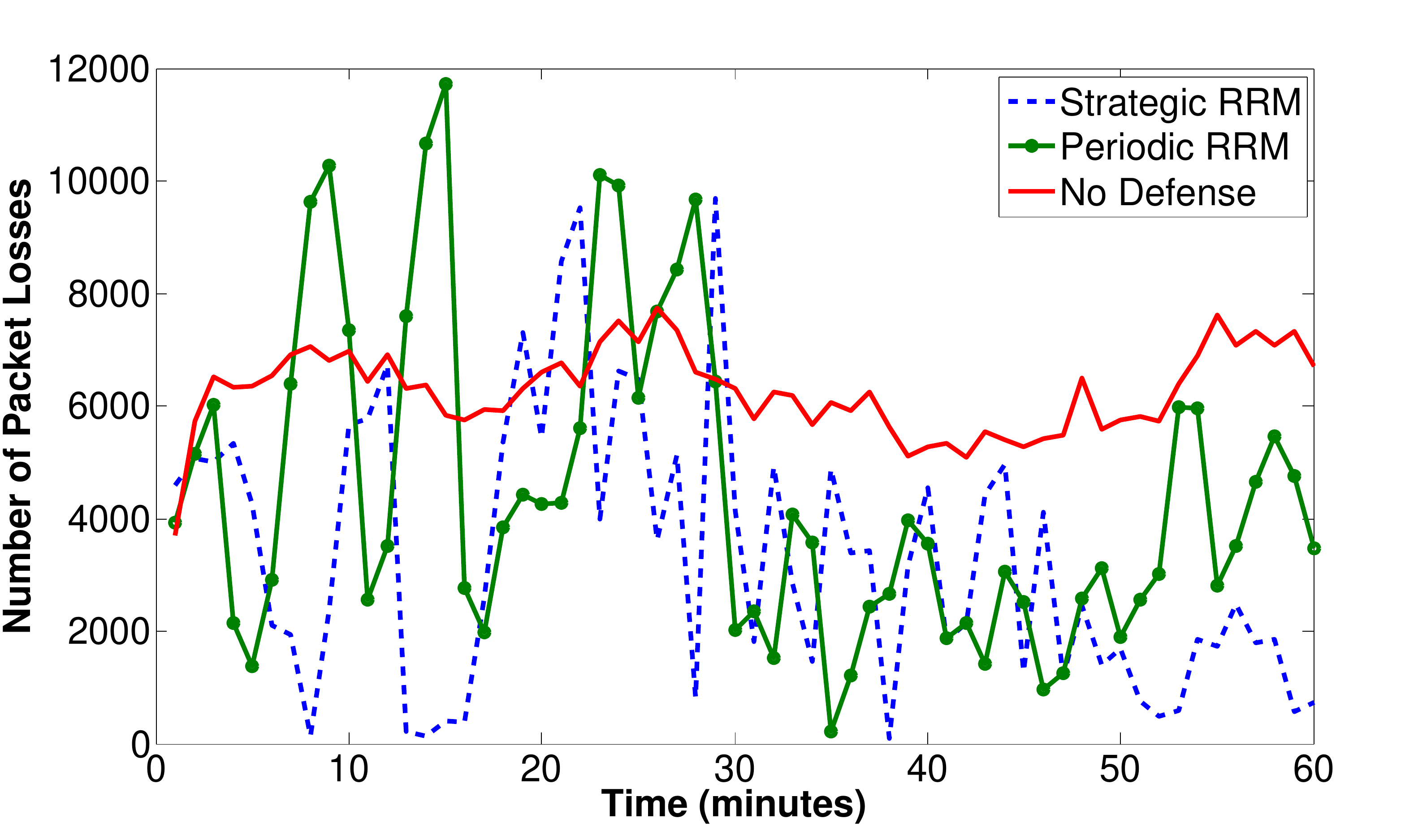}}
\subfloat[Stealthy Behavior Model \label{fig:powerful-ste}]{\includegraphics[width = 3.5in, height = 2in]{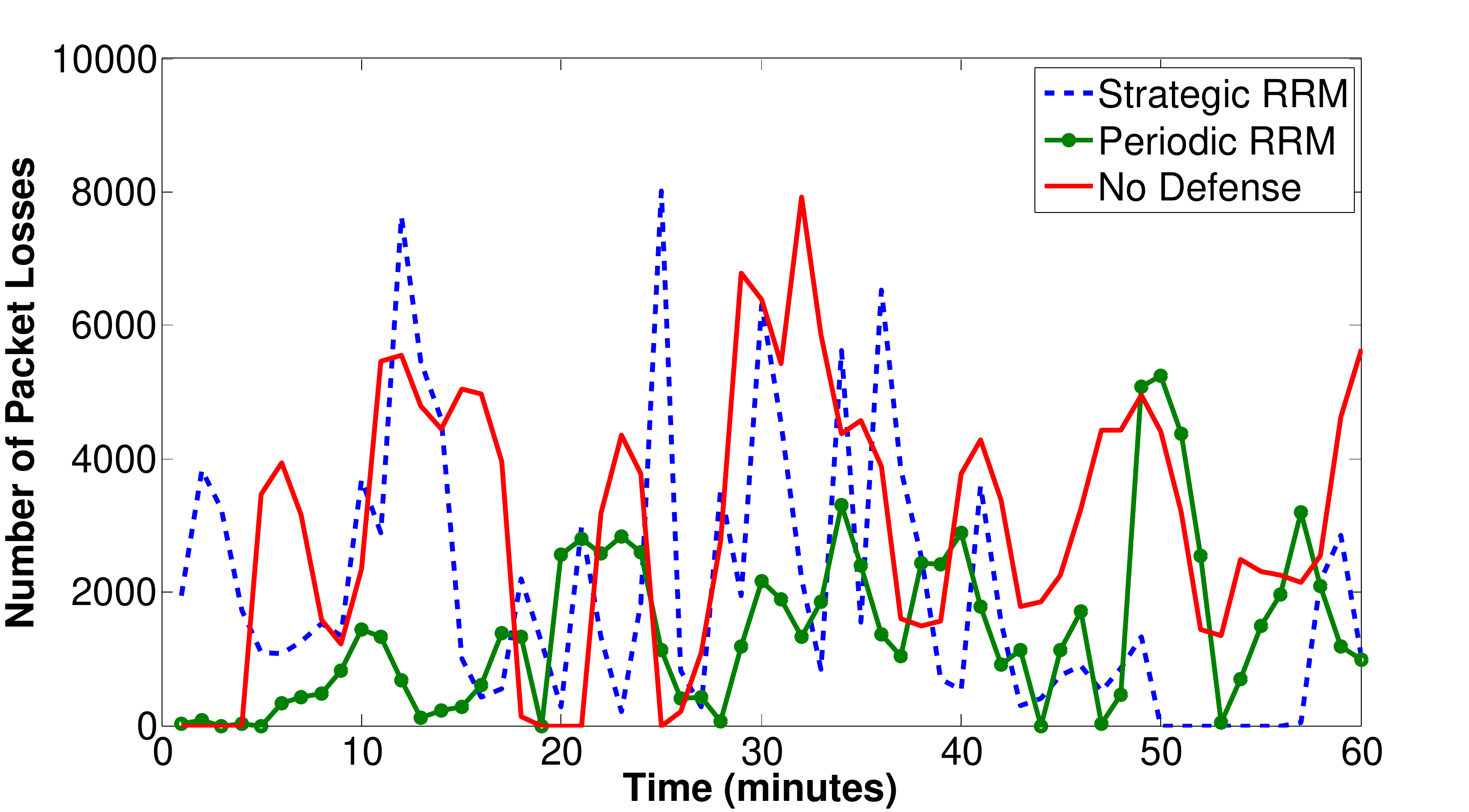}}
\vspace{-0.5em}
\caption{Packet Loss with Strong Attacker}
\label{fig:powerful}
\vspace{-2.5em}
\end{figure*}

\begin{figure*}[htb]
\setlength{\abovecaptionskip}{0.2cm}
\setlength{\belowcaptionskip}{0cm}
\centering 
\subfloat[Aggressive Behavior Model \label{fig:decent-agg}]{\includegraphics[width=3.5in,height =2in]{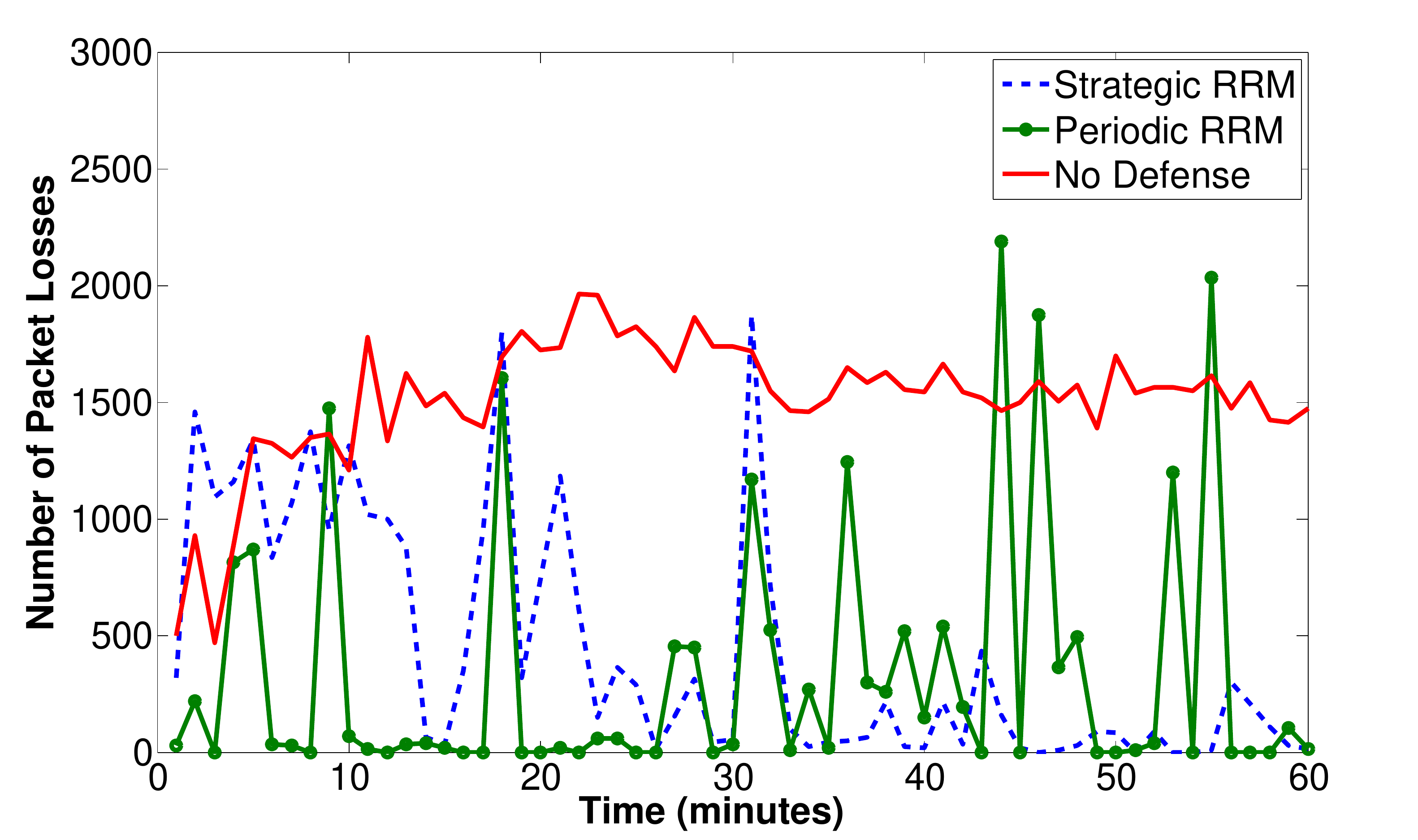}}
\subfloat[Stealthy Behavior Model \label{fig:decent-ste}]{\includegraphics[width = 3.5in, height = 2in]{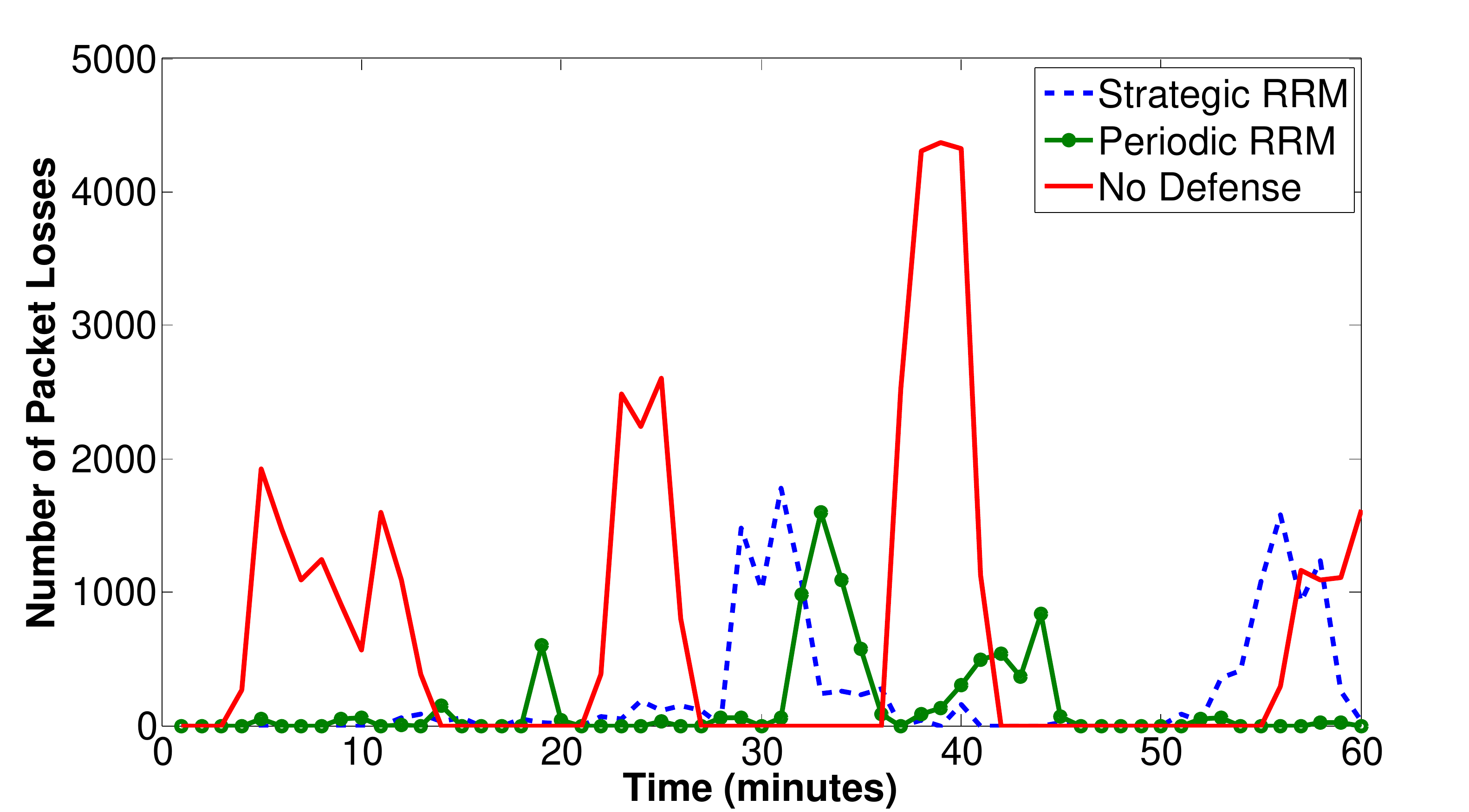}}
\vspace{-0.5em}
\caption{Packet Loss with Decent Attacker}
\label{fig:decent}
\vspace{-1.0em}
\end{figure*}


\subsubsection{Comparing Strategic RRM with Periodic RRM}

The values for different game parameters of \textit{Strategic RRM} are chosen as shown in Table~\ref{tab:parameters}. These parameters are observed and tuned based on our preliminary experiments on the RRM cost and the impact of attacks. More specifically, $\alpha$ is the attacker's gain, and it is considered as the percentage of packet losses due to the attack (when there is \textit{No Defense} mechanism), which is 3\% where the attacker is stealthy and decent, as shown in Table~\ref{tab:defendercost}. 
This value increases if the attacker's strategy is more intense or network size is smaller. 
Parameter $\lambda$ is the gain of the defender that s/he earns by defending against the reconnaissance attack. It is calculated by deducting 0.7\% packet loss when the 60-seconds \textit{Periodic RRM} is applied from 3\% packet loss when there is \textit{No Defense} mechanism. Thus, $\lambda$ is set at 2.3. 
$\lambda$ increases if RRM is applied more often than 60-seconds. 
The parameter $\delta$ is the cost of the defender with respect to the legitimate users. We derive this parameter as the percentage of increase in the average packet delay from the case of  \textit{No Defense} (a delay of 2175 microseconds) to the case of 60-seconds RRM (a delay of 2371 microseconds). Hence, we calculate $\delta$ as 9. 
The value of $\delta$ would increase if RRM is applied more often. 
In addition, we assign 1.6 to $c$ since it is the packet loss (as shown in Table~\ref{tab:switchingcost}) in the case of 60-seconds RRM.  
The cost variable, $c$, increases if RRM is more frequent.
Finally, we consider the minimum possible interval of running RRM as 1 second. Hence, the frequency ($f$) is set to 1.

\begin {table}[t]
\begin{center}
\caption {Parameters Used in Simulation} \label{tab:parameters} 
  \begin{tabular}{  | c  || c |  c | c | c | c |}
    \hline
      Parameter Name & $\alpha$ & $\lambda$ & $\delta$ & $c$ & $f$\\ \hline
    Value & 3  &   2.3  & 9 & 1.6 & 1  \\ \hline 
  \end{tabular}
\end{center}
\vspace{-2.5em}
\end {table}

In \textit{Strategic RRM}, the defender builds its belief about each of its clients of being a bot and perform RRM accordingly as it is mentioned in Section \ref{Sec:Belief}. In other words, the attacker is able to run a stealthy attack initially until the belief increases to some extent. Therefore, it is expected to see more packet losses for \textit{Strategic RRM} at the beginning of the experiments. 
Meanwhile, it is also expected that \textit{Periodic RRM} causes a similar number of packet losses since it has the same strategy throughout the experiments. 




In the next experiments, we implement a different attack models for each defender type and run them separately. 
%
In order to observe the performance in different attack models, we plot the graphs individually in Fig.~\ref{fig:powerful} and Fig.~\ref{fig:decent}.
Packet loss percentage and average delay for all network nodes are given in Table~\ref{tab:defendercost}.
The attacker model demonstrated in Fig.~\ref{fig:powerful-agg} 
considers a strong one with an aggressive behavior. As the Fig.~\ref{fig:powerful-agg} shows, in the case of \textit{No Defense}, the number of packet losses are high and it stays almost similar throughout the experiments. 
However, in the cases of \textit{Periodic RRM} and \textit{Strategic RRM} mechanisms, the scenario changes dramatically. \textit{Strategic RRM} has a much less number of packet losses after the first few minutes of the experiment. This initial period takes the belief to a state at which RRM starts to perform properly against the possible bots. An important observation is that \textit{Periodic RRM} usually has a higher number of packet losses continuously since it changes the routes at every time interval of the period which may not overlap with the attacker's reconnaissance phase. However, this is not the case for \textit{Strategic RRM}. It defends against the attack only when there are highly suspicious activities, which is reflected by the increased belief and the equilibrium conditions.  Another interesting observation in Fig. \ref{fig:powerful-agg} is that the peak number of packet losses appears at the beginning of the RRM cases, not in the \textit{No Defense} case.  
This can be explained as the nature of RRM. Since it is based on random mutation, it may cause more routes to the same link which can end up congesting that specific link intensively. However, in \textit{No Defense} case, the attacker can only target the same link set unless new bots are added and new routing paths are identified. In this attack model, the overall network's packet loss percentages are 31\% for \textit{No Defense}, 22\% for \textit{Periodic RRM}, and 16\% for \textit{Strategic RRM} as presented in Table \ref{tab:defendercost}. Average transmission delays are parallel to the packet loss results and decrease from 1172 milliseconds at \textit{No Defense} case to 989 milliseconds \textit{Periodic RRM} case, and to 438 milliseconds in \textit{Strategic RRM} case.

In addition to the attacker with an aggressive behavior, we also consider a stealthy one and represent the results in Fig. \ref{fig:powerful-ste}. In this graph, we can see unbalanced results which are basically due to the attacker's dynamic behavior. The stealthy attacker model can stay idle some of the time and can do longer reconnaissance attacks. Hence, it leads to such uneven consequences. Unlike the aggressive attack model, \textit{Strategic RRM} has higher packet losses (i.e., 10\%) compared to \textit{Periodic RRM} case's packet loss (i.e., 7\%), even after the belief gets time to be matured as shown in Table \ref{tab:defendercost}. The main reason behind this behavior is the characteristic of the stealthy attacker model where the attacker reduces and increases the attack intensity randomly making the defender's belief change more often. Even though \textit{Strategic RRM} has higher packet loss, the average transmission delay is 278 milliseconds for the \textit{Strategic RRM} case while it is 380 milliseconds in \textit{Periodic RRM}. Moreover, it should be noted that \textit{Periodic RRM} brings additional overhead to the network even if there is not an attack targeting the network. Hence, we claim that it is reasonable to use \textit{Strategic RRM} considering the overall advantages.   


Furthermore, we run the same experiments with a decent attacker. As we defined earlier, the number of bots he has is half of the bots that the aggressive attacker has. The results are shown in Fig.~\ref{fig:decent}, which are correlated with the ones in Fig. \ref{fig:powerful}. One significant observation is that the numbers of packet losses are decreased by more than half. In other words, increasing the number of bots, as in the aggressive case, raises the damage significantly, more than a linear increase. It is easily noticed that the number of packet losses goes down to zero in the case of the decent attacker as shown in Fig. \ref{fig:decent-agg} and Fig. \ref{fig:decent-ste}. This is because there are fewer bots to be protected against. The other result to pay attention is that there is a more high number of packet losses consecutively in \textit{Strategic RRM} compared to a strong attacker since less number of bots create less suspicion compared to more bots as it is shown in Equation~(\ref{Eq:belief}). This interesting outcome can be seen in the comparison of Fig. \ref{fig:powerful-ste} and Fig. \ref{fig:decent-ste}. In addition to that, \textit{Strategic RRM} has almost the same overall contribution (e.g., 2\% for aggressive, and 1\% for stealthy) to the network defense compared to \textit{Periodic RRM} (e.g., 2\% for aggressive, and 0.7\% for stealthy) as reported in Table \ref{tab:defendercost}. Yet, there is still a notable transmission delay decrease from 251 milliseconds to 200 milliseconds with Aggressive Attacker, and from 168 milliseconds to 141 milliseconds with Stealthy Attacker whenever \textit{Strategic RRM} is used. 



%% file: relatedwork.tex
\section{Related Work}
In this section, we talk about other's works in a similar area. We first explain research papers proposing defense mechanisms against SLFAs, and later we discuss signaling games paper for cybersecurity.

\vspace{-10pt}
\subsection{Defense Mechanisms for SLFA} 
There are various defense mechanisms proposed for the SLFA in literature. Only a few of them consider defending against reconnaissance phase of the attack \cite{hirayama2015fast} \cite{meier2018nethide} while most of them strive to protect the network during the data phase of the attack \cite{gkounis2016interplay} \cite{xue2018linkscope} \cite{wang2016towards} \cite{lee2013codef} \cite{wang2017linkbait} \cite{wang2018detecting}. 
Authors in \cite{hirayama2015fast} claim traceroute packets are increased before an SLFA occurs and they design a detection mechanism for the attacker based on that assumption. Even though we also have a similar belief that traceroute packets will increase before the attack, we do not rely on this feature individually and our attack mitigation technique is different. Meier et al. obfuscates the attacker's reconnaissance by replying with some virtual hops that is consistent with the network topology \cite{meier2018nethide}. While this solution aims to prevent the attacker to gain critical information, it gives wrong information to legitimate users. We argue that this approach is not the best practice since a legitimate user needs to get the right information in order to utilize the network facilities in the most efficient way. Otherwise, he might use longer path, and suffer higher transmission delays.

Meanwhile, authors in \cite{gkounis2016interplay} propose Traffic Engineering (TE) as a solution to mitigate an SLFA 
. In their solution, the defender forces the attacker to use the improbable path (i.e., very unlikely to be used) so that the attacker ends up being identified. Similarly, observing traffic patterns while attack happens is applied to detect and defer an SLFA in \cite{xue2018linkscope}. The main problem with these solutions is that they are reactive and some critical harm could have been done before the attack is mitigated. Traffic engineering based solution is also used by authors in \cite{wang2016towards}. The authors suggest upgrading switches to SDN-based switches in order to detect and mitigate an SLFA. However, it is not specified how to upgrade switches in run-time and how SDN switches are capable of detecting such attacks. In \cite{lee2013codef}, even though collaboration between ASes is shown helpful whenever an SLFA hits, it is not clear how to manage different ASes to work together. \cite{wang2017linkbait} observes link-probers by checking the packets at the ingress port. If a sender is found as link-prober, then Linkbait applies MTD to confuse its route. They suggest matrix-based feature extraction in order to detect which link-prober is a bot. In \cite{wang2018detecting}, authors propose utilizing SDN capabilities to detect the attacker's activities and block the malicious activity. While our solution is similar to \cite{wang2017linkbait} and \cite{wang2018detecting}, we do not rely on identifying bots which could not be possible with an intelligent adversary model. Differently, 
we have designed a signaling game which lets the defender to decide which action to take considering his/her payoff. 

\vspace{-10pt}
\subsection{Signaling Games for Cybersecurity} 

In many cybersecurity scenarios, the defender cannot clearly detect whether the received messages are coming from a benign user or they are a part of the attack scenario and are initiated by attackers. Signaling game is a two-player incomplete information game that has been used to model different cybersecurity problems, such as intrusion detection \cite{shen2011signaling} or deception \cite{zhuang2010modeling}. 
Furthermore, the authors model cybersecurity in terms of signaling games in \cite{casey2014cyber} and W. Casey et al. model deception with signaling game in~\cite{casey2018deception}. They present how signaling games provide a formal mathematical method to analyze the way of identity and deception coupling in cyber-social systems. The game-theoretic framework can be extended to reason about dynamical system properties and behavior traces. In \cite{mohammadi2016game}, the authors formulate a deception 
with signaling game in networks in which the defender deploys a fake avatar for identification of the compromised internal user. In~\cite{rahman2013game},  the authors propose a selective and dynamic mechanism for counter-fingerprinting. They model and analyze the interaction between a fingerprinter and a target as a signaling game. Following this work in \cite{zhao2017sdn}, the authors suggest changing attack surface (e.g., port numbers) depending on a belief that is observed in the signaling game. In~\cite{moghaddam2015trust}, the authors investigate the interactions between a service provider and a client by signaling game, where the client does not have complete information about the security conditions of the service provider. In~\cite{adili2017cost}, the authors propose a moving target-based deceptive defense mechanism using a signaling game for the frequency of migrations of the virtual machines in clouds. While we also apply signaling game for network defense, we propose using RRM which has not been studied under a signaling game. Furthermore, our framework is applied not only to reconnaissance attacks but also real data phase of the attacks. 
The reason behind choosing and applying the signaling game approach specifically for SLFA is that we believe the dynamic and unnoticeable behavior of SLFA can be modeled into a signaling game accurately. 



%% file: conclusion.tex
\section{Conclusion}
\label{sec:conclusion}

In this paper, we present a signaling game-based dynamic MTD to defend against Crossfire attacks. We first model the attacker and the defender as a signaling game. Considering their payoffs, we compute the equilibria of the game, which represent the best strategies for each player considering the opponent is rational. According to the game results, we develop an algorithm, namely \textit{Strategic RRM}. We implement and compare it with \textit{Periodic RRM}. Our experimental results show that \textit{Strategic RRM} can lessen the impact of Crossfire attacks similar to \textit{Periodic RRM}, while it brings significantly less overhead. 
As the future work, we 
would like to extend our strategic MTD-based framework for other DDoS-based emerging attacks.